\documentclass[11pt]{article}
\pdfoutput=1 

\PassOptionsToPackage{svgnames,table}{xcolor}
\usepackage{xcolor}

\usepackage[T1]{fontenc}

\usepackage[varqu]{zi4}

\usepackage[margin=15pt,font=small,labelfont={bf}]{caption}
\usepackage[margin=1in]{geometry}
\usepackage[english]{babel}

\usepackage[compact]{titlesec}
\usepackage{cmap}

\usepackage{bm}
\usepackage[colorinlistoftodos]{todonotes}
\usepackage{booktabs}
\usepackage{mathtools}

\usepackage{amsmath}
\usepackage{amssymb}
\usepackage{amsthm}
\usepackage{float}
\usepackage{xfrac}
\usepackage[numbers]{natbib}
\usepackage{hyperref}
\usepackage[most]{tcolorbox}
\hypersetup{colorlinks={true},urlcolor={blue},linkcolor={DarkBlue},citecolor=[named]{DarkGreen},linktoc=all}
\usepackage[section]{placeins} 
\newcommand{\BibTeX}{\rm B\kern-.05em{\sc i\kern-.025em b}\kern-.08em\TeX}
\usepackage{subcaption}
\usepackage{nicefrac}
\usepackage{pifont}

\usepackage{enumitem}
\usepackage{apxproof}
\usepackage{array,multirow,graphicx,bigdelim}
\input{insbox}
\usepackage{tikz}
\usepackage{tikz-dimline}
\usepackage{tikz-cd}
\usetikzlibrary{fit,calc,math,shapes,shapes.multipart,decorations.text,arrows,decorations.markings,decorations.pathmorphing,shapes.geometric,positioning,decorations.pathreplacing}

\tikzset{snake it/.style={decorate, decoration=snake}}

\pgfplotsset{compat=1.18}

\tikzset{
    leading_agent/.style={circle, draw={rgb, 255:red, 93; green, 166; blue, 13}, 
                          fill={rgb, 255:red, 232; green, 254; blue, 212}, 
                          line width=0.75pt, minimum size=6mm, inner sep=0pt}, 
    agent/.style={circle, draw=black, fill=lightgray!30, line width=0.75pt, minimum size=6mm, inner sep=0pt}, 
    etc/.style={draw=none, minimum size=6mm, inner sep=0pt},
    envy/.style={-{Triangle}, line width=1pt}, 
    pseudo/.style={-{latex}, dashed, draw=gray, line width=1pt},
    etc_edge/.style={-{Triangle}, dotted, line width=1pt},
    champ/.style={-{Triangle}, bend left=20, color={rgb, 255:red, 208; green, 2; blue, 27}, line width=1pt}, 
    best/.style={-{Triangle}, bend left=20, color=teal, line width=1pt},
    every loop/.style={min distance=12mm}
}

\colorlet{mygray}{gray!40}

\tcbset{
  importantdef/.style={
    colback=gray!7,
    colframe=black,
    boxrule=0.5pt,
    left=6pt,
    right=6pt,
    top=6pt,
    bottom=6pt,
    enhanced,
  }
}

\makeatletter

\let\oldnl\nl
\newcommand{\nonl}{\renewcommand{\nl}{\let\nl\oldnl}}
\makeatother

\usepackage{thmtools,thm-restate}
\usepackage[capitalise,nameinlink,noabbrev]{cleveref}

\theoremstyle{definition}
\newtheorem{definition}{Definition}

\newtheorem{lemma}{Lemma}
\newtheorem{theorem}{Theorem}
\newtheorem{corollary}{Corollary}
\newtheorem{proposition}{Proposition}

\newtheorem{claim}{Claim}

\Crefname{claim}{Claim}{Claims}
\Crefname{corollary}{Corollary}{Corollaries}
\Crefname{definition}{Definition}{Definitions}
\Crefname{example}{Example}{Examples}
\Crefname{lemma}{Lemma}{Lemmas}
\Crefname{property}{Property}{Properties}
\Crefname{proposition}{Proposition}{Propositions}
\Crefname{remark}{Remark}{Remarks}
\Crefname{theorem}{Theorem}{Theorems}

\newcommand{\M}{\mathcal{M}} 
\newcommand{\N}{\mathcal{N}} 
\newcommand{\V}{\mathcal{V}} 
\newcommand{\R}{\mathbb{R}_{{\ge}0}} 
\newcommand{\A}{\mathbb{A}}
\newcommand{\B}{\mathbb{B}}
\newcommand{\C}{\mathbb{C}}
\newcommand{\D}{\mathbb{D}}
\newcommand{\Pl}{\mathcal{P}}
\newcommand{\CGMchamp}{heavy-champion}
\newcommand{\ppa}{\textsc{3PA}^+}
\newcommand{\ppatypes}{\textsc{3PA}^+\textsc{-TYPES}}

\usepackage[linesnumbered,ruled,vlined]{algorithm2e}  

\SetCommentSty{mycommfont}

\theoremstyle{remark}

\makeatletter
\newcommand{\EFX}[1]{
  \if\relax\detokenize{#1}\relax
    \mathrm{EFX}
  \else
    #1\text{-}\mathrm{EFX}
  \fi
}
\makeatother

\newcommand{\EFXT}{\EFX{\frac{2}{3} }}
\newcommand{\EFXe}{\EFX{(1-\varepsilon)}}
\newcommand{\EFXa}{\EFX{\alpha}}

\title{Almost and Approximate EFX for Few Types of Agents}

\author{
	\begin{tabular}{m{0.12\textwidth}m{0.12\textwidth}m{0.12\textwidth}m{0.12\textwidth}%m{0.12\textwidth}m{0.12\textwidth}
 }
    \multicolumn{2}{c}{\textbf{Vishwa Prakash HV}} & \multicolumn{2}{c}{\textbf{Ruta Mehta}}
        \\
        \multicolumn{2}{c}{Chennai Mathematical Institute} & \multicolumn{2}{c}{University of Illinois Urbana-Champaign} %&
        \\ 
		    \multicolumn{2}{c}{\href{mailto:vishwa@cmi.ac.in}{\small{\texttt{vishwa@cmi.ac.in}}}} &\multicolumn{2}{c}{\href{mailto:rutameht@illinois.edu}{\small{\texttt{rutameht@illinois.edu}}}}
        \\
        &&&\\
		\multicolumn{4}{c}{\textbf{Prajakta Nimbhorkar}}\\
            \multicolumn{4}{c}{Chennai Mathematical Institute}\\
            \multicolumn{4}{c}{\href{mailto: prajakta@cmi.ac.cin}{\small{\texttt{prajakta@cmi.ac.in}}}}
	\end{tabular}
}

\DeclareMathOperator*{\argmax}{arg\,max}
\DeclareMathOperator*{\argmin}{arg\,min}

\date{}

\begin{document}

\maketitle

\begin{abstract}\label{sec:abstract}
    We study the problem of fair allocation of a set of indivisible goods among $n$ agents with $k$ distinct additive valuations, with the goal of achieving approximate envy-freeness up to any good ($\EFXa$).

    It is known that EFX allocations exist for $n$ agents when there are at most three distinct valuations \cite{hvetalEFXExists24}. \cite{amanatidisetalPushingFrontier24a} showed that a $\EFXT$ allocation is guaranteed to exist when number of agents is at most seven. In this paper, we show that a $\EFXT$ allocation exists for any number of agents when there are at most four distinct valuations.

    Secondly, we consider a relaxation called $\EFX{}$\textit{ with charity}, where some goods remain unallocated such that no agent envies the set of unallocated goods. \cite{akramietalEFXSimpler25} showed that for $n$ agents and any $\varepsilon \in \left(0, \frac{1}{2}\right]$, there exists a $\EFXe$ allocation with at most $\tilde{\mathcal{O}}((n/\varepsilon)^{\sfrac{1}{2}})$ goods to charity.  In this paper, we show that a $\EFXe$ allocation with a $\tilde{\mathcal{O}}(k/\varepsilon)^{\sfrac{1}{2}}$ charity exists for any number of agents when there are at most $k$ distinct valuations.
\end{abstract}

% We study the problem of fair allocation of a set of indivisible goods among $n$ agents with $k$ distinct additive valuations, with the goal of achieving approximate envy-freeness up to any good ($\alpha-\mathrm{EFX}$).

% It is known that EFX allocations exist for $n$ agents when there are at most three distinct valuations due to HV et al. Amanatidis et al. showed that a $\frac{2}{3}-\mathrm{EFX}$ allocation is guaranteed to exist when number of agents is at most seven. In this paper, we show that a $\frac{2}{3}-\mathrm{EFX}$ allocation exists for any number of agents when there are at most four distinct valuations.

% Secondly, we consider a relaxation called $\mathrm{EFX}$\textit{ with charity}, where some goods remain unallocated such that no agent envies the set of unallocated goods. \cite{akramietalEFXSimpler25} showed that for $n$ agents and any $\varepsilon \in \left(0, \frac{1}{2}\right]$, there exists a $(1-\varepsilon)-\mathrm{EFX}$ allocation with at most $\tilde{\mathcal{O}}((n/\varepsilon)^{\frac{1}{2}})$ goods to charity.  In this paper, we show that a $(1-\varepsilon)-\mathrm{EFX}$ allocation with a $\tilde{\mathcal{O}}(k/\varepsilon)^{\frac{1}{2}}$ charity exists for any number of agents when there are at most $k$ distinct valuations.
\section{Introduction}\label{sec:intro}

The problem of discrete fair division arises whenever a collection of \emph{indivisible} goods must be assigned to a group of agents, each with their own preferences over these goods. These agents may be people, autonomous software programs, or embodied agents such as robots. The goods in question are not divisible--there is no meaningful way to split a university seat or a domain name in half. In some settings, we are dividing public housing, school seats, or government resources among human participants. In others, we are allocating compute tokens, digital licenses, or scheduled access to systems among artificial agents. The underlying challenge in each of these settings is the same---how can we allocate scarce, indivisible goods in a fair manner?

The quintessential notion of fairness is that of envy-freeness \cite{foleyResourceallocation67, varianEquityenvy74}, which requires that each agent prefers the bundle assigned to them over that of any other agent. While compelling, this notion is often too strong in settings with indivisible goods, where such allocations may fail to exist even in simple cases. This has led to the study of relaxed fairness guarantees \cite{liptonetalapproximatelyfair04, caragiannisetalUnreasonableFairness19, plautroughgardenAlmostEnvyFreeness20, amanatidisetalMultiplebirds20} that retain the spirit of envy-freeness while accommodating the combinatorial nature of discrete settings. Among these, envy-freeness up to any good (\(\EFX{}\)) and its approximation, \(\EFXa\), have emerged as two of the central solution concepts.

An allocation is said to be \emph{envy-free up to any good} (\(\EFX{}\)) if each agent prefers the bundle assigned to them over any strict subset of the bundle assigned to any other agent. Does every instance of the discrete fair division problem admit an \(\EFX{}\) allocation? This seemingly simple question has turned out to be surprisingly difficult to answer. The answer to this question is known under a range of additional assumptions---such as restrictions on the number of agents, the structure of their valuation functions, or the number of unique valuation functions present in the instance. For example, \(\EFX{}\) allocations are known to exist when there are at most three agents \cite{chaudhuryetalEFXExistsJ.ACM24} with additive valuations\footnote{This result is known when two agents have arbitrary monotone valuations and one agent has {\em MMS-feasible} valuation \cite{akramietalEFXSimpler25}}. More recently, it has been shown that \(\EFX{}\) exists for three \emph{types} of agents with additive valuations \cite{hvetalEFXExists24}. That is, for any number of agents provided that the number of distinct valuation functions is at most three. For general valuations, \(\EFX{}\) allocations are known to exist for two agents \cite{plautroughgardenAlmostEnvyFreeness20} and more generally for two types of agents \cite{maharaExtensionAdditive24}.

In light of the challenges surrounding exact \(\EFX{}\) allocations, the literature has developed along several distinct directions---each relaxing the problem in a different way. One such direction considered a multiplicative approximation of \(\EFX{}\), known as \(\EFXa\), which requires that each agent prefers their bundle over any strict subset of the bundle assigned to any other agent up to a factor of \(\alpha\), for some \(\alpha\in (0,1]\). \citet{plautroughgardenAlmostEnvyFreeness20} showed the existence of \(\EFX{0.5}\) for any number of agents under general monotonic valuations. Under additive valuations, \citet{amanatidisetalMultiplebirds20}, showed that a \(\EFX{0.618}\) exists for any number of agents. More recently, \citet{amanatidisetalPushingFrontier24a} showed that a \(\EFXT\) allocation exists when there are at most seven agents with additive valuations. 

Another line of work \cite{chaudhuryetalLittleCharity21, bergeretalAlmostFullAAAI22, ghosaletalAlmostFullAAAI25} considers partial \(\EFX{}\) allocations, under an additional constraint that no agent envies the set of unallocated goods---often referred to as the \emph{charity}.  \citet{chaudhuryetalLittleCharity21} proved that for any number of agents with additive valuations, \(\EFX{}\) allocations always exists with at most \(n-1\) goods to the charity. \citet{bergeretalAlmostFullAAAI22} and \cite{maharaExtensionAdditive24} subsequently improved this bound independently, showing that an \(\EFX{}\) allocation always exists with at most \(n-2\) goods in the charity. Additionally, \cite{bergeretalAlmostFullAAAI22} proved that for four agents, an EFX allocation can always be found with at most one unallocated good. Recently, \citet{ghosaletalAlmostFullAAAI25} prove that an instance with \(n\) agents and \(k\) distinct valuation functions admits an \(\EFX{}\) allocation with at most \(k-2\) goods in the charity.

A third direction combines these two relaxations, seeking partial allocations that are approximately \(\EFX{}\) while minimizing the number of goods given to the charity. \citet{chaudhuryImprovingEFXGuarantees.EC.2021} proved that for any number of agents with additive valuations, and for any \(\varepsilon \in (0,\nicefrac{1}{2}]\) an \(\EFXe\) allocation exists with \(\mathcal{O}(\left( {n}/{\varepsilon}\right)^{\nicefrac{4}{5}})\) goods in the charity. This bound was later tightened by \citet{berendsohnetalFixedPointCycles22} to \(\mathcal{O}(\left( {n}/{\varepsilon}\right)^{\nicefrac{2}{3}})\), and subsequently improved to \(\tilde{\mathcal{O}}(\sqrt{n/\varepsilon})\)  by  \citet{akramietalEFXSimpler25}.

\paragraph{Instances with Few Types of Agents.} In many practical scenarios, it is reasonable to assume that agents can be grouped into a small number of types based on various shared characteristics or needs. For example, Consider the allocation of disaster relief resources. Here the beneficiaries can often be classified by demographic groups such as men and women, or by age brackets such as adults, children, and the elderly. In public health interventions, patients are grouped based on risk profiles. Similarly, in public housing assignments, applicants are frequently categorized into types based on household composition---such as single adults, families with children, or senior citizens---with each group having distinct preferences. These kinds of structured populations give rise to instances with a limited number of valuation types, although the total number of agents is large.

\paragraph{Our Contribution:} We consider the problem of \(\EFXa\) allocations for \(k\) types of agents with additive valuations. Our contributions are listed below:
\begin{itemize}
    \item In Theorem~\ref{thm:4types}, we show that for an instance with at most four types of agents, there always exists a complete \(\EFXT\) allocation.
    \item In Theorem~\ref{thm:rainbow-types}, we show that for an instance with at most \(k\) types of agents, there exists a \(\EFXe\) allocation with at most \(\tilde{\mathcal{O}}(\sqrt{k/\varepsilon})\) goods in the charity.
\end{itemize}
In the process of obtaining our first result, we discover a couple of errors in the work of \cite{amanatidisetalPushingFrontier24a} as discussed in Appendix~\ref{sec:error-in-3pa}. We first fix these, and then show the result for four types of agents.

\paragraph{Additional Related Work:}
We discuss some of the relevant literature here, and refer the reader to a comprehensive survey by \cite{amanatidis2023fair} for a more detailed literature review.

Apart from partial and/or approximate EFX allocations, an interesting constraint considered in literature is to restrict the number of possible values agents can have for the goods. In this context, \cite{amanatidis2021maximum} showed that EFX allocations exist for bi-valued instances, where each agent assigns one of two possible values to each good. This was extended to Pareto optimal EFX allocations by \cite{garg2023computing}. Existence of EFX allocations was shown by \cite{gorantlaetalFairallocation} for multi-set of goods, with the multiplicity as a parameter.

\section{Notations and Preliminaries}\label{sec:prelims}

We are given a fair division instance \(\mathcal{I}=\langle \N, \M, \V \rangle\), where \(\N\) is a set of \(n\) \emph{agents}, \(\M\) a set of \(m\) indivisible \emph{goods} and \(\V=\{v_i\}_{i \in \N}\) the \emph{valuation profile}, with each \(v_i:2^\M\to\R\) denoting agent \(i\)'s utility for any given set of goods. Throughout this paper we assume the valuations to be \emph{additive}. That is, $v_i(S) = \sum_{g \in S} v_i(g)$, for all \(i\in \N\), \(g\in \M\).

We use the term \emph{bundle} to denote a subset of goods. An \emph{allocation} \(X=\langle X_1,X_2,\ldots,X_n\rangle\) is a tuple of mutually disjoint bundles such that the bundle \(X_j\) is assigned to agent~\(j\). An allocation is said to be \emph{complete} if \(\bigcup_{j=1}^n X_j = \M\). An allocation that is not complete is called a \emph{partial allocation}. Given an allocation \(X\), The largest cardinality of any bundle in $X$ is called the {\em size} of $X$. Further, for a partial allocation \(X\), the set of unallocated goods \(\M\setminus\bigcup_{j=1}^n X_j\) is referred to as the \emph{pool}, denoted by \(\Pl(X)\).

Consider an agent \(a \in \N\), a good \(g \in \M\), and two subsets \(S, T \subseteq \M\). Let \(v_a\) denote the valuation function of agent \(a\). For brevity, we write \(v_a(g)\) as shorthand for \(v_a(\{g\})\), and use \(S \setminus g\) and \(S \cup g\) to denote \(S \setminus \{g\}\) and \(S \cup \{g\}\), respectively. We write \(S >_a T\) to mean that agent \(a\) strictly prefers \(S\) to \(T\), i.e., \(v_a(S) > v_a(T)\), and similarly use \(<_a\), \(\geq_a\), \(\leq_a\), and \(=_a\) for the corresponding comparisons.

\subsection{Fairness and Efficiency Notions: }\label{subsec:fairness_notions}
The fairness criteria relevant to our setting are introduced below. We start with the standard \emph{envy-freeness} notion, and then present some of the relaxations of it, namely  \(\EFXa\)  and \(\EFX{}\).

\begin{definition}[Envy, Envy-Freeness,  \(\EFXa\) ]
    Given an allocation \(X\), a constant \(\alpha\in (0,1]\), and two agents \(a,b\in \N\), we say: 
    \begin{itemize}
        \item Agent \(a\) does not \(\alpha\)-envy agent \(b\) or, equivalently, agent \(a\) is \(\alpha\)-envy-free towards agent \(b\) if \(v_a(X_a) \geq \alpha\cdot v_a(X_b)\).
        \item Agent \(a\) does not \(\alpha\)-envy agent \(b\) \emph{up to any good} or, equivalently, agent \(a\) is  \(\EFXa\)  towards agent \(b\) if \(v_a(X_a) \geq \alpha\cdot v_a(X_b\setminus h)\) for all \(h\in X_b\).
    \end{itemize}
    An allocation \(X\) is  \(\EFXa\)  if every agent is  \(\EFXa\)  towards every other agent. When \(\alpha=1\), we simply say that the allocation is \(\EFX{}\).
\end{definition}

Another relaxation of envy-freeness is \(\EFX{}\) with charity, which allows some goods to remain unallocated, as long as no agent envies the set of unallocated goods. More formally, it is defined as follows:

\begin{definition}[$\EFXa$ with Charity, \cite{chaudhuryetalLittleCharity21}]
    A partial allocation $X$ is said to be $\EFXa$ with charity $\mathcal{P}(X)$, if $X$ is an $\EFXa$ allocation in which no agent envies the set $\mathcal{P}(X)$. That is, for every agent $a\in \N$, $v_a(X_a)\ge v_a(\mathcal{P}(X))$.
\end{definition}

In the context of EFX allocations, \citet{chaudhuryetalEFXExistsJ.ACM24} showed that one can assume without loss of generality that the input instances are \emph{non-degenerate}, meaning that no agent values two different bundles equally. This assumption simplifies the analysis and is often used in the literature. Throughout this paper, we too assume the instances to be non-degenerate. We recall this definition below.

\begin{definition}[Non-degenerate Instance, \cite{chaudhuryetalEFXExistsJ.ACM24,akramietalEFXSimpler25}]\label{def:non-degeneracy}
    An instance \(\langle \N, \M, \V \rangle\) of fair allocation is called non-degenerate if no agent values two different bundles equally. That is, for any agent \(a\in \N\) and any two bundles \(S, T\subseteq \M\), if \(v_a(S)=v_a(T)\), then \(S=T\).
\end{definition}

They show that for any positive integer \(n\), if every non-degenerate instance of fair division on \(n\) agents admits EFX allocations, then every the instances, including degenerate ones, admit EFX allocations. We note that this result also holds for  \(\EFXa\)  allocations. For the sake of completeness, we provide a proof of this result in Appendix~\ref{sec:non-degeneracy-assumption}.

We next define the \emph{Pareto optimality} criterion, which is a standard efficiency notion, as well as a measure of progress of iterative procedures in the fair allocation literature.
\begin{definition}[Pareto Domination and Optimality]
    Given two allocations \(X\) and \(Y\), we say that \(X\) \emph{Pareto dominates} \(Y\), denoted by \(X\succ Y\), if for every agent \(a\in \N\), \(v_a(X_a) \geq v_a(Y_a)\) and there exists at least one agent \(b\in \N\) such that \(v_b(X_b) > v_b(Y_b)\). An allocation is said to be \emph{Pareto optimal} if it is not Pareto dominated by any other allocation.
\end{definition}

\paragraph{Envy Graphs: }
An important tool in the study of fair allocation of indivisible goods is the  notion of \emph{envy graphs}, which captures the envy relationships among agents in a given allocation. We now define various types of envy graphs used in this paper. In these definitions, an allocation may be complete or partial.

\begin{definition}[Envy Graph \(G\)]
    Given an allocation \(X\), the envy graph \(G(X) = (\N, E(X))\), is a directed graph on the set of agents. There exists a directed edge \((a,b)\) if and only if agent \(a\) envies agent \(b\), that is, \(E(X)=\{\,(a,b)\mid v_a(X_a) < v_a(X_b)\}\).
\end{definition}

\begin{definition}[Reduced Envy Graph \(G_r\) \cite{amanatidisetalPushingFrontier24a}]
    Given an allocation \(X\), the \emph{reduced envy graph on \(X\)}, denoted by \(G_\mathrm{r}(X) = (\N, E_\mathrm{r}(X))\), is a sub-graph of the envy graph \(G(X)\), obtained by removing each edge \((a, b)\), such that \(|X_a|>1\), \(|X_b|=1\), and \(a\) is \(\frac{2}{3}\)-envy-free towards \(b\). 
    That is \(E_\mathrm{r}(X) = E(X) \setminus \{\,(a, b) \mid |X_a|>1, |X_b|=1, \text{ and } v_a(X_a) \ge \frac{2}{3} v_a(X_b)\}\).
\end{definition}

\begin{definition}[Enhanced Envy Graph \(G_e\) \cite{amanatidisetalPushingFrontier24a}]
    Given an allocation \(X\), the \emph{enhanced envy graph on \(X\)}, \(G_\mathrm{e}(X) = (\N, E_\mathrm{e}(X))\), is a super-graph of the reduced envy graph \(G_\mathrm{r}(X)\), obtained by adding edges of the form \((a, s)\), where \(s\) is a source (i.e., vertex of in-degree 0) in \(G_\mathrm{r}(X)\), \(|X_a|=1\), \(|X_s|>1\), and \(v_a(X_s) \geq \frac{2}{3} v_a(X_a)\). These edges are referred to as the \emph{red edges}. 
\end{definition}

Given an allocation, the \textsc{Envy-Cycle-Elimination} (ECE) procedure of \citet{liptonetalapproximatelyfair04} iteratively finds a directed cycle in the envy graph and cyclically permutes the bundles along the edges of the cycle. This strictly increases the utility of agents in on the cycle and does not decrease anyone's utility, resulting in a Pareto dominating allocation. The process repeats until the envy graph becomes acyclic. Importantly, if the initial allocation is \(\alpha\)-EFX, then every intermediate and the final allocation produced by ECE is also  \(\EFXa\) \cite{plautroughgardenAlmostEnvyFreeness20}, and the final envy graph is acyclic. Moreover, \citet{amanatidisetalPushingFrontier24a} show that the ECE procedure maintains \(\EFXT\) when run on either the extended or the reduced envy graphs. In Section~\ref{sec:corrected-proof-ece}, we identify a small error in the proof of termination of the ECE procedure in \cite{amanatidisetalPushingFrontier24a} and provide a corrected proof.

\paragraph{Critical Goods: }\citet{markakissantorinaiosImprovedEFX23} showed that, for any \(\beta\in(0,1)\) and any \(\alpha\in(0,1)\), given a partial  \(\EFXa\) allocation in which no agent has a good in the pool that she values by atleast a \(\beta\)-factor of her own bundle, (referred to as a {\em critical good}, defined below), then completing this partial allocation using the \textsc{Envy-Cycle-Elimination} procedure (allocate a good to the source and eleminate cycle) yields a complete \(\min\left(\alpha,\frac{1}{\beta +1}\right)\)-EFX allocation. We recall this result here.

\begin{definition}[Critical Good, \cite{markakissantorinaiosImprovedEFX23}]
    Given a partial allocation \(X\), a good \(g\in \Pl(X)\) is called \(\beta\)-\emph{critical} good for an agent \(a\in \N\) if \(a\) prefers \(g\) by a factor of at least \(\beta\). That is, \(v_a(g) > \beta\cdot v_a(X_a)\).
\end{definition}

\begin{lemma}[\cite{markakissantorinaiosImprovedEFX23}]\label{lemma:markakis-santorinaios}
    Given a partial  \(\EFXa\)  allocation \(X\) in which no agent has a \(\beta\)-critical good in the pool, there exists a complete \(\EFX{\min\left(\alpha,\frac{1}{\beta +1}\right)}\) allocation \(Y\) such that \(Y\succ X\).
\end{lemma}

In this work, one of our primary interests is in finding complete \(\EFXT\) allocations. Therefore, unless stated otherwise, we will assume \(\alpha = \frac{2}{3}\) and \(\beta = \frac{1}{2}\). Consequently, when we refer to a \emph{critical good}, we mean a \(\frac{1}{2}\)-critical good by default. 

\subsection{Instances with \(k\) Types of Agents}
Our main focus is on instances with \(k\) distinct valuations, called as the \(k\)-type instances. Such instances are considered in literature, e.g. in \cite{maharaExtensionAdditive24, hvetalEFXExists24}. We recall some properties of such instances from \cite{hvetalEFXExists24}. An instance \(\langle \N, \M, \V \rangle\) of fair allocation is called a {\em \(k\)-type instance} if the agents in \(\N\) can be partitioned into \(k\) parts, such that all agents in the same part have the same valuation function. That is, \(\N=\N^1\uplus\N^2\cdots\N^k\),  where \(\N^i=\{a_1^i, a_2^i,\ldots, a_{p_i}^i\}\) represents the set of all agents who have the identical valuation function \(v_i\). We call \(\N^i\) a \emph{group} of agents.

\paragraph{Leading Agents and the Ordering Invariant:} Given an allocation \(X\) for a \(k\)-type instance, we denote the bundle assigned to agent \(a_i^t \in \N^t\) as \(X_i^t\). Furthermore, we assume, without loss of generality,  that within each group, the bundles are sorted in non-decreasing order of their valuations. That is, for any two agents \(a_j^t, a_\ell^t\in \N^t\), if \(j<\ell\), then \(v_t(X_j^t) \le v_t(X_\ell^t)\). This is referred to as the \emph{ordering invariant}. As a consequence, for any \(t\in [k]\), agent \(a_1^t\) is the agent with the least valued bundle in \(\N^t\), and is referred to as the \emph{leading agent} of the group \(\N^t\).

Given a \(k\)-type instance, we denote the set of leading agents as \(\mathcal{L} = \{a_1^1, a_1^2, \ldots, a_1^k\}\). The leading agent of group \(\N^t\) is denoted by \(a_1^t\). The following proposition captures some properties of \(k\)-type instances and partial allocations, which will be useful in our proofs.

\begin{proposition}[Properties of $k$-types Instances and Partial Allocations \cite{hvetalEFXExists24}]\label{prop:k-types-envy-properties}
Let $\langle \N, \M, \V \rangle$ be a $k$-types instance, and let $X$ be a partial allocation in which every agent receives at least one good. Then the following hold:
\begin{enumerate}
    \item The envy graph $G(X)$ has at most $k$ sources, and every source is a leading agent of its group.
    \item For any group $\N^t$ and any non-leading agent $a_j^t \in \N^t$ ($j > 1$), if $a_j^t$ envies a bundle $S \subseteq \M$, then the leading agent $a_1^t$ also envies $S$.
    \end{enumerate}
\end{proposition}

The following corollary is a simple observation that follows from the ordering invariant:
\begin{corollary}\label{cor:critical}
    For any group $\N^t$ and any non-leading agent $a_j^t \in \N^t$ ($j > 1$), if a good $g \in \Pl(X)$ is $\beta$-critical for $a_j^t$, then $g$ is also $\beta$-critical for the leading agent $a_1^t$.
\end{corollary}

\subsection{Background on rainbow cycle number and result of \cite{chaudhuryImprovingEFXGuarantees.EC.2021}}

Our second result deals with the problem of finding a \(\EFXa\) allocation with bounded charity. \citet{chaudhuryImprovingEFXGuarantees.EC.2021} give a relation between the problem of finding an approximate \(\EFX{}\) allocation with bounded charity and the  rainbow cycle number. Here, we recall some of the relevant results and definitions from \cite{chaudhuryImprovingEFXGuarantees.EC.2021}.

\begin{definition}[Rainbow cycle number \cite{chaudhuryImprovingEFXGuarantees.EC.2021}]\label{def:rainbow}
    Given an integer $d>0$, the rainbow cycle number $R(d)$ is the largest $k$ such that
there exists a $k$-partite graph $G=(V_1\cup V_2 \cup \ldots\cup V_k,E)$ with the following properties:
\begin{itemize}
    \item each part has at most d vertices, i.e., $|V_i|\leq d$, and

    \item every vertex has at least one incoming edge from every part other than the one containing it,
and
\item there exists no cycle $C$ in $G$ that visits each part at most once.
\end{itemize}
\end{definition}
We refer to cycles that visit each part at most once as ``rainbow'' cycles.

The following theorem relates the rainbow cycle number with approximate \(\EFX{}\) allocations with bounded charity:
\begin{theorem}[Theorem 4, \cite{chaudhuryImprovingEFXGuarantees.EC.2021}]
    Let \(h(d)=d\cdot R(d)\) and \(\varepsilon\in (0,\frac{1}{2}]\). Let 
    \(h^{-1}\big(\frac{n}{\varepsilon}\big)\) be the smallest integer such that \(h(d)\geq \frac{n}{\varepsilon}\). Then there is a \((1-\varepsilon)\)-EFX allocation \(X\) and a set of unallocated goods $\mathcal{P}(X)$ such that 
    \(|\mathcal{P}(X)|=O\big(\frac{4n}{\varepsilon\cdot h^{-1}(\frac{2n}{\varepsilon})}\big)\), and no one envies $\mathcal{P}(X)$.
\end{theorem}
\citet{chaudhuryImprovingEFXGuarantees.EC.2021} show that $R(d)\in O(d^4)$ and thus show the existence of a $\EFXe$ allocation with $O((\frac{n}{\varepsilon})^{4/5}$ charity.
The following theorem from \cite{akramietalEFXSimpler25} gives an almost tight upper bound on the rainbow cycle number:
\begin{theorem}[Theorem 3, \cite{akramietalEFXSimpler25}]
Given any integer $d>0$, the rainbow cycle number $R(d)\in O(d\log d)$.
\end{theorem}

This implies the existence of $(1-\varepsilon)$-EFX allocations with $\tilde{O}((\frac{n}{\varepsilon})^{1/2})$ charity. 
\section{The \(\ppa\) Algorithm of \cite{amanatidisetalPushingFrontier24a}}\label{sec:ppa}

Our algorithm for computing a \(\EFXT\) allocation builds on some of the techniques introduced in \cite{amanatidisetalPushingFrontier24a}. In this section, we briefly describe the algorithm and its properties.

Given an instance of fair division, if there are at least as many agents as the number of goods, then an allocation of size at most one, is trivially EFX. Therefore, one can safely assume that there are more goods than agents. That is, \(m>n\). To begin with, consider a partial allocation with exactly one good per agent. If there is an envy-cycle in the allocation, resolve it using the ECE procedure, and make the envy graph acyclic. This allocation is called a \emph{seed} allocation. \citet{amanatidisetalPushingFrontier24a} proposed an way to compute a \emph{Property Preserving Partial Allocation}, via an algorithm referred to as {\em the \(\ppa\) algorithm} (Algorithm~\ref{alg:3PA_plus}). Some of the details regarding the subroutines of the algorithm is deferred to Appendix~\ref{sec:ppa-algorithm}. The algorithm takes a seed allocation as input and returns a (possibly partial) allocation that satisfies a specific set of properties, listed below:

Desired properties of a partial allocation \(X\):
    \begin{enumerate}
        \item Every agent \(a \in \N\) with \(|X_a|=1\) is EFX towards any other agent. \label{prop:a}
        \item Every agent \(a \in \N\) is \(2/3\)-EFX towards any other agent. \label{prop:b}
        \item Every agent \(a \in \N\) with \(|X_a|>1\) does not have any critical goods, i.e., for every agent \(a \in \N\) and any good \(g \in \mathcal{P}(X)\), \(|X_a|>1 \Rightarrow v_a(g)\leq \frac12 v_a(X_a)\).\label{prop:d}
        \item Any agent \(a\) with \(|X_a|=1\) has at most one critical good \(g_a\), and she values that good at most \(\frac23\) of the value of her current bundle, i.e., for every agent \(a\in\N\), there is at most one good \(g_a \in \mathcal{P}(X)\) such that \(v_a(g_a)> \frac12 v_a(X_a)\) and it holds that \(v_a(g_a) \le  \frac23 v_a(X_a)\).\label{prop:e}      
    \end{enumerate}

\begin{algorithm*}
    
    \DontPrintSemicolon
    \SetNoFillComment
    \LinesNotNumbered 
    \caption{\textsc{\sc Property-Preserving Partial Allocation (3PA+)} \cite{amanatidisetalPushingFrontier24a}} \label{alg:3PA_plus}
    \SetKwComment{Comment}{/* }{ */}
    \SetKw{Continue}{continue}
    \SetKw{Break}{break}
    \SetKw{Step}{Step}
    \SetKwData{Kw}{}
    \KwIn{An instance \((\N,\M,\V)\), and a partial allocation \(X\) of size at most \(2\) which satisfies Properties~\ref{prop:a} and \ref{prop:b}.}
    \KwOut{A Property-Preserving Partial Allocation \(X^1\) of size at most \(2\) such that either \(X^1\) is complete, or \(G_e(X^1)\) has at least one source and every source in \(G_e(X)\) has exactly two goods. \vspace{3pt}}
    
    \While{\(\mathcal{P}(X) \neq \emptyset\)}
    {
    \nl \label{step1}
    \uIf{there is \(a \in \N\) with \(|X_a|=1\) and a good \(g \in \mathcal{P}(X)\) such that \(v_a(g) > v_a(X_a)\)}{
    \(\mathcal{P}(X) \gets (\mathcal{P}(X)\cup \{X_a\})\setminus \{g\}\) and \(X_a \gets \{g\}\)\; \tcp{If an agent with 1 good prefers 1 good from the pool, swap them.}}

    \nl \label{step2}
    \uElseIf{ \(\exists a \in \N\) with \(|X_a|=2\) and a good \(g \in \mathcal{P}(X)\) such that \(v_a(g) > \frac{3}{2}v_a(X_a)\)}{
    \(\mathcal{P}(X) \gets (\mathcal{P}(X)\cup \{X_a\})\setminus \{g\}\) and \(X_a \gets \{g\}\) \;
    \tcp{Else if an agent with 2 goods prefers 1 good from the pool by more than \(3/2\), swap her bundle with that good.}
    }
    
    \nl  \label{step3}
    \uElseIf{ \(\exists a \in \N\) with \(|X_a|=1\) and  \(g_1,g_2 \in \mathcal{P}(X)\) with \(v_a(\{g_1,g_2\}) > \frac{2}{3}v_a(X_a)\)}{
    \(\mathcal{P}(X) \gets (\mathcal{P}(X)\cup \{X_a\})\setminus \{g_1,g_2\}\) and \(X_a \gets \{g_1,g_2\}\)\;
    \tcp{Else if an agent with 1 good prefers 2 goods from the pool by more than \(2/3\), swap that 1 good with the two goods from the pool.}
    }
    
    \nl \label{step4}
    \uElseIf{\(\exists a \in \N\) with \(|X_a|=2\), and goods \(g \in \mathcal{P}(X)\) and \(g' \in X_a\) s.t \(v_a(g) > v_a(g')\)}
    { \(\mathcal{P}(X) \gets (\mathcal{P}(X) \cup \{g'\}) \setminus \{g\}\) and \(X_a \gets (X_a \cup \{g\}) \setminus \{g'\}\).\;
    \tcp{Else if an agent with 2 good prefers 1 good from the pool to one of her own goods, swap that good with the one good from the pool.}
    }
    
    \nl \label{step5}
    \uElseIf{the reduced envy graph \(G_r(X)\) has cycles}{\(X \gets\textsc{AllCyclesResolution}(X,G_r)\)\;
    \tcp{Else if the reduced envy graph has cycles, resolve them by swapping the bundles.}}

    \nl \label{step6}
    \uElseIf{in the reduced envy graph \(G_r\) there is a \emph{source} \(s\) with \(|X_s|=1\)}
    { \(\mathcal{P}(X) \gets \mathcal{P}(X)\setminus \{g^*\}\), \(X'_s \gets X_s \cup \{g^*\}\), where \(g^* \in \argmax_{g \in \mathcal{P}(X)}v_a(g)\).\;
    \tcp{Else if there is a source in the reduced envy graph with a single good, add her most valuable good from the pool to her bundle.}
    }
    
    \nl  \label{step7}
    \uElseIf{\(|\{g \in \mathcal{P}(X): \exists i \text{ such that } |X_a|=1 \text{ and } v_a(g) > \frac{2}{3}v_a(X_a)\}| = 1\)}
    {\(X \gets \textsc{SingletonPool}(X)\)\;
    \tcp{Else if there is a single unallocated good and some agent with \(1\) good prefers it strictly more than \(2/3\), run the \textsc{SingletonPool} subroutine to allocate it.}}
    
    \nl \label{step8}
    \uElseIf{the enhanced envy graph \(G_e(X)\) has cycles}{\(X \gets \textsc{AllCyclesResolution}(X,G_e)\)\;
    \tcp{Else if the enhanced envy graph has cycles, resolve them by swapping the bundles.}
    }

    \nl \label{step9}
    \uElseIf{There exists a path \( \Pi =(s,\ldots,a)\) in \(G_e(X)\) starting at a source \(s\) of \(G_r(X)\), such that \(v_a(X_a)<v_a(\{g,g'\})\) for some \(g\in \mathcal{P}(X)\) and some \(g'\in X_s\)}
    {\(X \gets \textsc{PathResolution*}(X, G_e, \Pi)\)\;
    \tcp{Else if there exists a path from a source \(s\) to some agent \(a\) in \(G_e(X)\), such that \(a\) prefers a good from the pool and a good from \(X_s\) to her own bundle, swap the bundles along the envy path and give those two goods to \(a\).}
    }
    
    \nl \label{step10}
    \Else {\Break\;}
    }
    \Return \(X\) \; 
    
\end{algorithm*}

The algorithm consists of a while loop that iteratively applies a sequence of nine steps. Each step is executed only if the  partial allocation satisfies its associated condition. The loop terminates when none of the nine conditions are met, at which point, the algorithm returns a partial $\frac{2}{3}$-EFX allocation satisfying all the properties listed earlier, proved in \cite{amanatidisetalPushingFrontier24a}. They also prove the following lemma:

\begin{lemma}[\cite{amanatidisetalPushingFrontier24a}]\label{lemma:3PA_plus}
    Given an input allocation \(X\) in which each agent has exactly one good and the envy graph is acyclic, The \(\ppa\) algorithm (Algorithm~\ref{alg:3PA_plus}) outputs a partial allocation \(Y\) of size  at most \(2\) that satisfies Properties~\ref{prop:a} to \ref{prop:e}. Furthermore, if \(Y\) is not a complete allocation, then \(G_e(Y)\) has at least one source, and every source \(s\) in \(G_e(Y)\) has \(|{Y}_s|=2\).
\end{lemma}

After computing such a partial allocation \(X\), then the next step is to allocate all the critical goods that are in the pool, to build another partial \(\EFXT\) allocation, say \(Y\). Finally, the allocation is completed iteratively allocating the remaining goods to the source and resolving the envy-cycles using the ECE procedure. Due to result of \citet{markakissantorinaiosImprovedEFX23}, this procedure preserves the \(\EFXT\) property.

In the \(\ppa\) algorithm, one of the steps is to execute the the ECE procedure on the extended envy graph \(G_e\). Since this is shown to preserve the Properties~\ref{prop:a} to \ref{prop:e}, the following corollary holds.
\begin{corollary}\label{cor:bundles_along_edges}
    Consider a \(\EFXT\) allocation \(X\). If \((a,b)\) is an edge in the graph \(G_e(X)\), then in any allocation with the same bundles as in \(X\), if agent \(a\) has the bundle \(X_b\), then she is \(\EFXT\) towards all the agents. Furthermore, if \(|X_b|=1\), then with \(X_b\), agent \(a\) is \(\EFX{}\) towards all.
\end{corollary}

\section{\(\EFXT\) for Four Types of Agents} \label{sec:4types}

In this section, we show that every instance with at most four types of agents admits a \(\EFXT\) allocation. 

Firstly, we modify the \(\ppa\) algorithm to accommodate for four types of agents by introducing additional steps, Steps~\ref{step9.1}-\ref{step9.5}, inserted immediately after Step~\ref{step9}. The modification is presented as Algorithm~\ref{alg:3PA_plus_type}. In these steps, we invoke a new subroutine, called as the \textsc{Pseudo-Cycle-Resolution} (Procedure~\ref{alg:pseudo-path-resolution}), detailed below.

\begin{algorithm}
    \DontPrintSemicolon
    \SetNoFillComment
    \LinesNotNumbered 
    \caption{\textsc{\sc Additional step to get \(\text{3PA}^+\)-types} \cite{amanatidisetalPushingFrontier24a}} \label{alg:3PA_plus_type}
    \SetKwComment{Comment}{/* }{ */}
    \SetKw{Continue}{continue}
    \SetKw{Break}{break}
    \SetKw{Step}{Step}
    \SetKwData{Kw}{}
    \KwIn{As in \(\text{3PA}^+\)}
    \KwOut{Same as \(\text{3PA}^+\)}
    \SetKwProg{Dots}{\ldots}{ \ldots}{\ldots}
    \vspace{0.5cm}
        \Dots{(same as \textnormal{\(\text{3PA}^+\)})\,\medskip}
    {

        \nlset{9.1}\label{step9.1}
        \uIf{There is a single source, say $d_1$ wlog, \textbf{\emph{and}} \(|\D|\ge 2\),   }{
            \nlset{9.2}\label{step9.2}
            \uIf{\(|X_{a_1}|=|X_{b_1}|=|X_{c_1}|=1\) \textbf{\emph{and}} \({(}|X_{d_2}|=2 \textbf{\emph{ or }} (d_1,d_2)\notin G_e(X){)}\)}{

                \nlset{9.3} \Kw{\texttt{Step 9.3}}\label{step9.3}
                \uIf{some leading agent in $\A\cup\B\cup\C$, say $a_1$, envies $d_2$}{
                    Let $\pi$ be a leading-path from $d_1$ to $a_1$ \tcp*{Such a path exists because of Lemma~\ref{lem:leading-path}}
                    \textsc{Pseudo-Cycle-Resolution}($\pi, d_2, X_{d_2}$)\;
                }
                \;
                \nlset{9.4} \Kw{\texttt{Step 9.4}}\label{step9.5}
                \uIf{A leading agent in $\A\cup\B\cup\C$, say $a_1$ wlog, has $v_a(X_{a_1})<v_a(\{g,g'\})$, where $g\in X_{d_2}$ and $g'\in \mathcal{P}(X)$}{

                    Let $\pi$ be a leading-path from $d_1$ to $a_1$ \tcp*{Such a path exists because of Lemma~\ref{lem:leading-path}}
                    \textsc{Pseudo-Cycle-Resolution}($\pi, d_2$, $\{g, g'\}$)\;
                    \(\mathcal{P}({X}) \gets (\mathcal{P}(X) \cup X_{d_2}) \setminus \{g,g'\}\)
                }

            }
        }
    }
    \vspace{0.5cm}
    \dots\\
    \Return \(X\) \; 
\end{algorithm}

\SetAlgorithmName{PROCEDURE}{Subroutine}

\begin{algorithm}
    \DontPrintSemicolon
    \caption{\textsc{Pseudo-Cycle-Resolution} (\(\pi\), \(z^*\), S)} \label{alg:pseudo-path-resolution}
    \SetKwComment{Comment}{/* }{ */}
    \KwIn{A path $\pi = (z_1, z_2, \ldots, z_\ell)$ in \(G_e(X)\), an agent $z^*$, a bundle \(S\), and an allocation $X$}
    \KwOut{Updated allocation $Y$}
    Initialize $Y \gets X$\;
    \For{$j = 2$ \KwTo $\ell-1$}{
        $Y_{z_j} \gets X_{z_{j+1}}$\;
    }
    $Y_{z^*} \gets X_{{z}_2}$\; 
    $Y_{{z}_\ell} \gets S$\;
    
    Reorder bundles among agents within each group in $Y$ to maintain the ordering invariant\;
    \Return{$Y$}
\end{algorithm}
\SetAlgorithmName{ALGORITHM}{Algorithm}

\paragraph{Intuition for additional steps:} 

In \cite{amanatidisetalPushingFrontier24a}, they first compute a partial \(\EFXT\) allocation \(X\) satisfying properties \ref{prop:a} to \ref{prop:e}, using the \(\ppa\) algorithm. The next is to allocate critical goods, and then apply Lemma~\ref{lemma:markakis-santorinaios} to get a complete  \(\EFXT\)  allocation. They consider at most $7$ agents. A good that is critical for only one agent is allocated to a source and envy-cycles, if any, are resolved. Since each agent has at most one critical good, and source has no critical goods, there can be at most three goods that are critical for more than one agent, called {\em contested critical goods}. They show that these goods can be allocated to the source(s) in the enhanced envy graph without violating the  \(\EFXT\)  property. This also uses another crucial fact that, when there are three contested critical goods, there is only one source, and moreover, every agent except the source has a critical good.

In our setting, we have an arbitrary number of agents with at most four different valuations. Hence each group can have at most one critical good, except the group containing a source. Thus there can be at most three critical goods. If the enhanced envy graph has at least two sources, the critical items can be allotted as in \cite{amanatidisetalPushingFrontier24a}. However, when there is only one source, and three critical goods (one per group), then even other agents in the same group as the source have no critical goods. Thus the strategy of \cite{amanatidisetalPushingFrontier24a} of allocating all the critical goods to the source does not work. This necessitates additional steps (Steps \ref{step9.1} to \ref{step9.5}) in the $\ppa$ algorithm. We call our algorithm the $\ppatypes$.

\begin{algorithm}
    \DontPrintSemicolon
    \caption{\(\textsc{FewTypesAllocate}\)(\(\langle\N,\M,\V \rangle\))} \label{alg:few_types}
    \SetKwComment{Comment}{/* }{ */}
    \KwData{An additive instance with at most \(4\) types of agents.}
	\KwResult{A \(2/3\)-EFX allocation \(X\).\vspace{3pt}}
    Let \(X^0\) be an arbitrary partial allocation where each bundle has cardinality \(1\) and \(G(X^0)\) is acyclic.\;
    \(X \gets \ppatypes(\langle\N,\M,\V \rangle, X^0)\)\;
    \(\mathcal{C} \gets \{g\in \Pl(X) \mid  g \text{ is critical for some agent }\}\)\;

    \BlankLine
	\If{\label{alg:case1}\(|\mathcal{C}| = 2\) \text{ and there are at least two sources } \(s_1, s_2\) \text{ in } \(G_e(X)\)}{
		Suppose \(\mathcal{C} = \{g_1, g_2\}\)\;
		\(Y_{s_1} \gets X_{s_1} \cup \{g_1\}\)\;
		\(Y_{s_2} \gets X_{s_2} \cup \{g_2\}\)\;
		\(Y_i \gets X_i\) for all \(i \in \N \setminus \{s_1, s_2\}\)\;
	}

	\ElseIf{\label{alg:case2}(\ \(|\mathcal{C}| = 2\) \textbf{\emph{and}} there is only one source in \(G_e(X)\), say $d_1$\ ) \textbf{\emph{or}} \(|\mathcal{C}| = 1\)}{
		
		\(Y_{d_1} \gets X_{d_1} \cup \mathcal{C}\)\;
		\(Y_i \gets X_i\) for all \(i \in \N \setminus \{d_1\}\)\;
	}
    \BlankLine
	\Else(\hfill\tcp*[f]{That is, \(|\mathcal{C}| = 3\)}) {	\label{alg:case3}
		Let \(d_1\) be a source in \(G_e(X)\)\tcp*{At least one leading agent who is a source in \(G_e(X)\)}
		\(h \gets \argmin_{g\in \mathcal{C}} v_{d}(g)\)\tcp*{least valued good om $\mathcal{C}$ according to $v_d$}
		
		\If(\tcp*[f]{\(d_1\) is the only agent in \(\D\) }){\(|\D| = 1 \) }{
				\(Y_{d_1} \gets X_{d_1} \cup \mathcal{C}\)\;
		}

		\Else(\hfill\tcp*[f]{That is, \(\D = \{d_1, d_2, \ldots\}\)}){
				\If{\(v_{d}(X_{d_2}) < \frac{2}{3} v_{d}(X_{d_1} \cup \mathcal{C} \setminus \{g\})\) for some \(g \in X_{d_1} \cup \mathcal{C}\)}{
					\(Y_{d_2} \gets X_{d_2} \cup \{g\}\)\;
					\(Y_{d_1} \gets X_{d_1} \cup (\mathcal{C} \setminus \{g\})\)\;
				}
				\Else{
					\(Y_{d_1} \gets X_{d_1} \cup \mathcal{C}\)\;
					\(Y_{d_2} \gets X_{d_2}\)\;
				}
		}
		\(Y_i \gets X_i\) for any other \(a\in \N\setminus \{d_1,d_2\}\)\;
	}

    \(X \gets \textsc{EnvyCycleElimination}(Y, G(Y))\)
    \tcp*{Complete the allocation}
    \Return \(X\)
\end{algorithm}

In the additional steps, whenever it is not possible to allot all three critical items to the source, one of the items needs to be allotted to another agent from the same group. But if this agent is already envied by someone, allocating additional item to that agent breaks the \(\EFX{}\) property. This needs a careful resolution of the envies, prior to allotment of critical items, as given in Procedure~\ref{alg:pseudo-path-resolution}.

{\bf Bottleneck for more number of types:} When there are five or more types of agents, there are more than three critical items. Allocating a critical item to a source and then resolving cycles, like that in the subroutine for {\em uncontested critical} in \cite{amanatidisetalPushingFrontier24a} creates potential strong envies from other agents from the same group, as the same item may be critical to more than one agent in the same group.

\paragraph{Pseudo-Cycle-Resolution (Procedure~\ref{alg:pseudo-path-resolution}):} The usual ECE procedure resolves cycles among agents that is formed by envy edges. Pseudo-cycles are formed when an agent (and therefore a leading agent, say $x_1$) from a group, say $\mathbb{X}$, has a path of envy-edges to a leading agent in another group, say , say $y_1$ of $\mathbb{Y}$. However, if no one in $\mathbb{Y}$, and hence $y_1$, in particular, envies $x_1$'s bundle, then there is no envy-cycle. Now, consider the situation when $y_1$ envies a bundle of another agent in group $\mathbb{X}$, say that of $x_i$. Although $x_1,\ldots,y_1, x_i$ is not a cycle, bundles can be exchanged along every edge in this path, preserving \(\EFX{}\) or $\EFXT$ property. Since we invoke this procedure only on allocations of size at most \(2\) in Steps~\ref{step9.3} and \ref{step9.5}, the correctness of this procedure is later shown in Lemma~\ref{lemma:ppa_types}. See \cref{fig:pseudo-path-resolution} for an illustration of the \textsc{Pseudo-Path-Resolution} procedure.

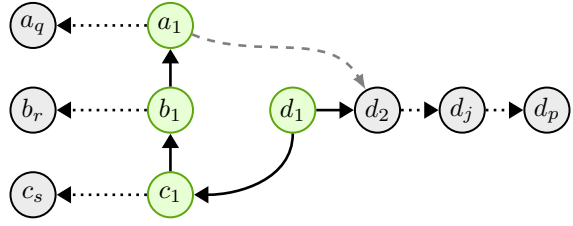
\begin{figure}
    \centering
    \begin{tikzpicture}
    % Nodes
    \node[leading_agent] (d1) {\small \(d_1\) };
    \node[agent, right=0.5cm of d1] (d2) {\small \(d_2\) };
    \node[agent, right=0.5cm of d2] (dj) {\small \(d_j\) };
    \node[agent, right=0.5cm of dj] (dp) {\small \(d_p\) };
    \node[leading_agent, left=1cm of d1] (b1) {\small \(b_1\) };
    \node[leading_agent, above=0.5cm of b1] (a1) {\small \(a_1\) };
    \node[leading_agent, below=0.5cm of b1] (c1) {\small \(c_1\) };
    \node[agent, left=1.2cm of a1] (aq) {\small \(a_q\) };
    \node[agent, left=1.2cm of b1] (br) {\small \(b_r\) };
    \node[agent, left=1.2cm of c1] (cs) {\small \(c_s\) };

    % Edges
    \draw[pseudo] (a1) to[out=-22.5, in=120] (d2);
    \draw[etc_edge] (a1) -- (aq);
    \draw[etc_edge] (b1) -- (br);
    \draw[etc_edge] (c1) -- (cs);
    \draw[etc_edge] (d2) -- (dj);
    \draw[etc_edge] (dj) -- (dp);
    \draw[envy] (d1) -- (d2);
    \draw[envy] (d1) to[out=-90, in =0] (c1);
    \draw[envy] (c1) -- (b1);
    \draw[envy] (b1) -- (a1);  
\end{tikzpicture}
    \caption{Illustration of the \textsc{Pseudo-Cycle-Resolution} procedure. The path $\pi = (d_1, c_1, b_1, a_1)$. Agent $a_1$ prefers one good $g$ from $X_{d_j}$ and one good $g'$ from the pool over her own bundle $X_{a_1}$. The \textsc{Pseudo-Path-Resolution} procedure reallocates as follows: $X_{c_1}$ is assigned to $d_j$, $X_{b_1}$ to $c_1$, $X_{a_1}$ to $b_1$, and $\{g, g'\}$ to $a_1$. The bundles are then reordered within each group to maintain the ordering invariant.}
    \label{fig:pseudo-path-resolution}
\end{figure}

In \cite{amanatidisetalPushingFrontier24a}, it is shown that all the steps of the $\ppa$ algorithm maintain properties \ref{prop:a} and \ref{prop:b} in each iteration, and when the $\ppa$ algorithm terminates, properties \ref{prop:a} to \ref{prop:e} are satisfied. We will show that the output of the \(\ppatypes\) also satisfies these properties.

In a run of the \(\ppatypes\) algorithm, Step~\ref{step9.3} or \ref{step9.5} is applied only when the allocation fails all the conditions of Steps ~\ref{step1} to \ref{step9}, and if the enhanced envy graph has a single source, say \(d_1\), and further, the other leading agents have bundles of size \(1\). In this case, we first show that the leading agents have a path from \(d_1\) such that the path contains no non-leading agents.  

\begin{restatable}{lemma}{leadingpath}\label{lem:leading-path}
Let $X$ be a partial allocation of size at most two, satisfying properties \ref{prop:a} to \ref{prop:e}. If $G_e(X)$ is acyclic, with a unique source say $d_1$ wlog, and all other leading agents have bundles of size $1$, then every leading agent has a path from $d_1$ such that the path contains no non-leading agents. 
\end{restatable}
\begin{proof}
    Recall that \(\mathcal{L}\) is the set of all leading agents. Consider an arbitrary leading agent other than \(d_1\), say \(a_1 \in \mathcal{L} \setminus \{d_1\}\). Since \(a_1\) is not a source in \(G_e(X)\), it must have an incoming edge from some agent, say \(u_j\), from some group \(\mathbb{U}\). Furthermore, since \(a_1\) has a singleton bundle, all incoming edges must be envy edges (not red edges). Therefore $u_j$ cannot be in the same group as $a_1$. By Proposition~\ref{prop:k-types-envy-properties}, if \(u_j\) envies \(a_1\), then the corresponding leading agent \(u_1\) also envies \(a_1\). Since \(u_1\) also has a singleton bundle, the edge \((u_1, a_1)\) is present in \(G_e(X)\). Thus, every leading agent in \(\mathcal{L} \setminus \{d_1\}\) has an incoming envy edge from some leading agent in \(\mathcal{L}\).

    Now, consider the subgraph \(H\) of \(G_e(X)\) induced by \(\mathcal{L}\). The graph \(H\) is a single-source acyclic graph: every leading agent in \(\mathcal{L} \setminus \{d_1\}\) has an incoming edge, and since \(G_e(X)\) is acyclic, so is \(H\). Therefore, every vertex in \(H\) is reachable from the source \(d_1\) via a path consisting only of leading agents. This completes the proof.

\end{proof}

Now, to prove that the \(\ppatypes\) algorithm satisfies properties \ref{prop:a} to \ref{prop:e}, it suffices to show that Steps~\ref{step9.1} to \ref{step9.5}, preserve these properties.

\begin{restatable}{lemma}{Ppatypes}
    \label{lemma:ppa_types}
    Consider a \(4\)-types instance. If properties \ref{prop:a} and \ref{prop:b} are satisfied by an allocation before the execution of Step~\ref{step9.1} of $\ppatypes$ algorithm, then they are satisfied after the execution of Steps~\ref{step9.1} to \ref{step9.5}.
\end{restatable}
\begin{proof}
    Since Steps~\ref{step9.1} to \ref{step9.5} are the last steps in the sequence in $\ppatypes$, if any of these steps is executed on \(X\), then $X$ must have failed all the conditions from Step~\ref{step1} to Step~\ref{step9}. Step~\ref{step9.3} only shuffles bundles, whereas Step~\ref{step9.5} creates a new bundle of size $2$. So no bundle has size more than $2$. 
    Further, as Step~\ref{step9.3} shuffles existing bundles, every agent not involved in the pseudo-cycle remains \(\EFX{}\)  (or $\EFXT$, as the case may be, before the execution of these steps) towards every other agent, including those on the pseudo-cycle. 

    Without loss of generality, consider a path \(\pi\) in the extended envy graph, from a source, say \(d_1\) to a leading agent, say \(a_1\). Because of Lemma~\ref{lem:leading-path}, we can assume that this path involves only leading agents. If $b_1$ and/or $c_1$ are on $\pi$, they get an envied bundle of the same size i.e. $1$ and hence trivially remain \(\EFX{}\)  towards all the agents apart from $a_1$, who might have received a newly formed bundle. Also, as agent $a_1$ gets a higher valued bundle than $X_{a_1}$, she remains \(\EFX{}\)  towards others. Thus, to complete the proof, we need to show the following:
    \begin{enumerate}
        \item\label{itm:a1} For any agent \(x\) in \(\N\setminus d_2\), if \(x\) has a singleton bundle, then she is \(\EFX{}\)  towards $a_1$, and if \(x\) has a bundle of size $2$, then she is $\EFXT$ towards $a_1$.
        \item\label{itm:d2} Agent $d_2$ is \(\EFX{}\)  towards every agent. 
    \end{enumerate}

    {\em Proof of \ref{itm:a1}:} If Step~\ref{step9.3} is executed, then this is trivial, as no new bundle is formed. So consider the case when Step~\ref{step9.5} is executed. That is, agent \(a_1\) gets a good \(g'\) from \(X_{d_1}\) and a good \(g''\) from the pool. Since this step is executed only if Step~\ref{step9.3} fails, no agent except possibly $d_1$, envies $d_2$. Therefore, no agent in \(\N\setminus \{d_1,d_2\}\) envies the good \(g'\). Furthermore, no agent values a good in the pool more than their own bundle, as Steps~\ref{step1} \ref{step4} are not met. So for any agent in \(\N\setminus \{d_1,d_2\}\), any one of the two goods in $a_1$'s new bundle is less valuable than her own bundle. So all of them are \(\EFX{}\)  towards $a_1$. 
    
    If Step~\ref{step9.5} is executed with $|X_{d_2}|=2$, then by Property~\ref{prop:b}, $d_1$ is $\EFXT$ towards $d_2$ before this step. Hence for each item $g$ in $X_{d_2}$, we have $v_d(X_{d_1})\geq \frac{2}{3}v_d(g)$. Moreover, since Step~\ref{step4} fails, for any item $g'$ in $\mathcal{P}(X)$, we have $v_d(g')\leq \frac{1}{2}v_d(X_{d_1})$. Therefore $d_1$ is $\EFXT$ towards $a_1$. If Step~\ref{step9.5} is executed with $|X_{d_2}|=1$, then by Condition~\ref{step9.2}, $v_d(X_{d_2})=v_d(g)\leq \frac{3}{2}v_d(X_{d_1})$. Hence again $d_1$ is $\EFXT$ towards $a_1$. This completes Part~\ref{itm:a1}.

    {\em Proof of \ref{itm:d2}:} The new bundle that $d_2$ gets is from a leading agent, say $u_1$. Hence it is a singleton bundle by Condition~\ref{step9.2}. So, to satisfy all the properties, $d_2$ must be \(\EFX{}\)  towards everyone else. Since $d_1$ is a source and Step~\ref{step6} fails, we have $|X_{d_1}|=2$. As the envy edge from $d_1$ to $u_1$ exists in $G_e(X)$, it must be the case that $v_d(X_{u_1})>\frac{3}{2}v_d(X_{d_1})$. Therefore, for any bundle $X_{v_j}$ in $X$, $d_1$ is $\EFXT$ towards $X_{v_j}$ implies that $d_2$, with the new bundle $Y_{d_2}=X_{u_1}$ is \(\EFX{}\)  towards $v_j$. It remains to show that $d_2$ is \(\EFX{}\)  towards $a_1$. 

    If Step~\ref{step9.3} is executed with $|X_{d_2}|=2$, then $d_1$ is $\EFXT$ towards $d_2$ before this step, and $v_d(X_{u_1})>\frac{3}{2}v_d(X_{d_1})$ together imply that $d_2$ is \(\EFX{}\)  towards $Y_{a_1}=X_{d_2}$. If Step~\ref{step9.3} is executed with $|X_{d_2}|=1$, then by Condition~\ref{step9.2}, $(d_1,d_2)$ is not present in $G_e(X)$. However, $(d_1,u_1)$ is present in $G_e(X)$. This implies that $d_2$ finds the new bundle $Y_{d_2}=X_{u_1}$ more valuable than her old bundle $X_{d_2}$, and hence is envy-free towards $a_1$.

    Now consider Step~\ref{step9.5}, and hence $Y_{a_1}=\{g,g'\}$. By Condition~\ref{step9.2}, Step~\ref{step9.5} is executed only if $|X_{d_2}|=2$ or $(d_1,d_2)$ edge doesn't exist in $G_e(X)$. 
    
    Let $|X_{d_2}|=1$, and hence $(d_1,d_2)$ is not present in $G_e(X)$. Since Step~\ref{step1} fails, $v_d(g')\leq v_d(X_{d_2})<v_d(Y_{d_2})$. Also, $v_d(g)=v_d(X_{d_2})<v_d(Y_{d_2})$. So $d_2$ is \(\EFX{}\)  towards $a_1$.

    Let $|X_{d_2}|=2$. Then $d_1$ is $\EFXT$ towards $a_1$ with the new bundle $Y_{a_1}=\{g,g'\}$, and $v_d(Y_{d_2})>\frac{3}{2}v_d(X_{d_1})$ implies that $d_2$ is \(\EFX{}\)  towards $a_1$.
\end{proof}

In the \textsc{FewTypesAllocate} algorithm, once an allocation \(X\) is obtained from the \(\ppatypes\) algorithm, the critical goods in \(\Pl(X)\) are allocated. Finally the non-critical goods using the procedure of \citet{markakissantorinaiosImprovedEFX23}. Before we proceed to the proof of Theorem~\ref{thm:4types}, we establish some properties of the allocation \(X\) obtained from the \(\ppatypes\) algorithm.

\begin{restatable}{lemma}{twothird}\label{lemma:2/3}
    Let \(X\) be a partial allocation on a four-types instance, obtained from the \(\ppatypes\) algorithm. Let \(s\) be a source in the enhanced envy graph \(G_e(X)\). If any agent \(a_i^\ell\) in a group \(\N^\ell\) has a critical good in \(\Pl(X)\), then for every agent \(a_j^\ell\in \N^\ell\), the value of the bundle \(X_s\) is less than two third the value of their own bundle. That is,  \(\ v_\ell({X}_s) < \frac{2}{3}v_\ell({X}_j^\ell)\).
\end{restatable}
\begin{proof}
    Suppose some agent \(a_j^\ell\) has a critical good \(g\) in \(\Pl(X)\). Then, by Proposition~\ref{prop:k-types-envy-properties}, \(g\) is also critical for the leading agent \(a_1^\ell\). As \(X\) is obtained from the \(\ppatypes\) algorithm, from Property~\ref{prop:d}, we know that \(|{X}_1^\ell|=1\). 

    Now, since \(s\) is a source in \(G_e(X)\), Applying Lemma~\ref{lemma:3PA_plus}, we know that \(|{X}_s|=2\). Since \(G_e(X)\) is a super-graph of \(G_r(X)\), the node \(s\) must also be a source in \(G_r(X)\). If \(\ v_\ell({X}_s) \ge \frac{2}{3}v_\ell({X}_1^\ell)\), then there must have been a red edge from \(a_1^\ell\) to \(s\) in \(G_e(X)\). This contradicts the fact that \(s\) is a source in \(G_e(X)\). Therefore, \(\ v_\ell({X}_s) < \frac{2}{3}v_\ell({X}_1^\ell)\). Since \(\forall a_j^\ell\in \N^\ell\), \(v_\ell(X_j^\ell) \ge v_\ell(X_1^\ell)\), we have \(\ v_\ell({X}_s) < \frac{2}{3}v_\ell({X}_j^\ell)\).
\end{proof}

Furthermore, the output of the \(\ppa\) algorithm on a four-types instance satisfies the following lemma.

\begin{lemma}\label{lemma:same-good-critical}
    Let \(X\) be an allocation obtained from the \(\ppa\) algorithm. Let \(g,g'\in \Pl(X)\) be the critical goods for agents \(a\) and \(b\) respectively. If \(a\) and \(b\) belong to the same group, then \(g\) and \(g'\) are one and the same.
\end{lemma}
\begin{proof}
    Suppose \(a\) and \(b\) belong to the same group, say \(\N^\ell\). Then, by Proposition~\ref{prop:k-types-envy-properties}, \(g\) and \(g'\) are also critical for the leading agent \(a_1^\ell\). Since \(X\) is obtained from the \(\ppa\) algorithm, from Properties~\ref{prop:d} and \ref{prop:e}, we know that \(a_1^\ell\) can have at most one critical good. Therefore, the goods \(g\) and \(g'\) are one and the same.
\end{proof}

Finally, assuming the termination of the algorithm, Theorem~\ref{thm:4types} shows that the \textsc{FewTypesAllocate} algorithm (Algorithm~\ref{alg:few_types}) computes a \(\EFXT\) allocation on an instance with at most four types of agents. The termination is proved later in Lemma~\ref{lem:termination}.

\begin{restatable}{theorem}{FourTypesThm}\label{thm:4types}
    Given an instance \(\langle \N,\M,\V\rangle\) of at most four types of agents, when the Algorithm~\ref{alg:few_types} \textsc{FewTypesAllocate} terminates, it outputs a \(\EFXT\) allocation.
\end{restatable}
\begin{proof}
    From \cref{lemma:3PA_plus} and \cref{lemma:ppa_types}, we know that the \(\ppatypes\) algorithm (Algorithm~\ref{alg:3PA_plus_type}) returns a partial allocation \(X\) that satisfies Properties~\ref{prop:a} - \ref{prop:e}. Furthermore, \(G_e(X)\) has at least one source, and every source has exactly two goods. If \(X\) is a complete allocation, then Algorithm~\ref{alg:few_types} outputs that allocation, thus we are done.
    
    Consider the case when $X$ is not a complete allocation. Then, \(\mathcal{C}\) be the set of goods in \(\Pl(X)\) that are critical for at least one agent. That is, \(\mathcal{C} = \{g\in \Pl(X) \mid g \text{ is critical for some agent }\}\).

    From Lemma~\ref{lemma:same-good-critical}, we know that each group can have at most one critical good. Therefore, there can be at most four critical goods in total. However, from Lemma~\ref{lemma:3PA_plus}, we know that \(G_e(X)\) has at least one source, and the source cannot have a critical good. From Proposition~\ref{prop:k-types-envy-properties}, we know that the source must be a leading agent of some group, say \(\N^1\). Then, because of Proposition~\ref{prop:k-types-envy-properties}, no agent from the group \(\N^1\) has any critical goods. Therefore, there can be at most \(3\) critical goods, at most one for each of the remaining groups.

    We consider three cases based on the number of sources in \(G_e(X)\) and the number of critical goods: 

    \noindent\paragraph{Case 1: $|\mathcal{C}|\leq2$, and $G_e(X)$ has two sources:} Let $\mathcal{C}=\{g_1,g_2\}$, and let \(s_1\) and \(s_2\) be the sources in $G_e(X)$. This is considered in Line~\ref{alg:case1} of Algorithm~\ref{alg:3PA_plus_type}.  Here,  \(Y\) is constructed from \(X\) by allocating \(g_1\) to \(s_1\) and \(g_2\) to \(s_2\).

    Consider an agent \(u\in\N\) who has a critical good, say \(g_1\in \mathcal{C}\) under the partial allocation \(X\). Notice that \({u}\) cannot be a source vertex, due to Property~\ref{prop:e}. Since \(X\) is the output of the \(\ppatypes\) algorithm, by applying Property~\ref{prop:e}, we have \(v_{u }(g_1)\le \frac{2}{3}v_{u}({X}_{u})\) and \(v_{u }(g_2)\le \frac{2}{3}v_{u}({X}_{u})\). Also, by Lemma~\ref{lemma:2/3}, \(v_{u}({X}_{s_j}) < \frac{2}{3}v_{u}({X}_{u})\) for \(j \in \{1, 2\}\). Therefore, for \(j\in \{1,2\}\), we have 
    \begin{align}
        v_{u}(Y_{s_j}) &= v_{u}({X}_{s_j} \cup g_j)\notag\\
                             &< \frac{2}{3}v_{u}({X}_{u}) + \frac{2}{3}v_{u}({X}_{u})\notag\\
                             &< \frac{3}{2}v_{u}({X}_{u})\notag\\
                             &= \frac{3}{2}v_{u}(Y_{u})\notag
    \end{align}
      
    This agent \(u\) is \(\EFXT\) towards both \(s_1\) and \(s_2\). Since all other agents retain their old bundles, agent \(u\) continues to be \(\EFXT\) towards them.
    
    Now, consider an agent \(u\ \in \N\) who does not have a critical good under the partial allocation \(X\). Then, it is clear that for each \(j\in\{1,2\}\), we have \(v_{u}(g_j) \le \frac{1}{2}v_{u}({X}_{u}) \le \frac{1}{2}v_{u}(Y_{u})\). Furthermore, since \(s_1\) and \(s_2\) are sources in \(G_e(X)\), we have \(v_{u}(X_{s_j}) < v_{u}({X}_{u})\), for each \(j\in \{1,2\}\). Therefore, 
    
    \begin{align}
        v_a(Y_{s_j}) &= v_{u}({X}_{s_j} \cup g_j)\notag\\
                             &< v_{u}({X}_{u}) + \frac{1}{2}v_a({X}_{u})\notag\\
                             &\le \frac{3}{2}v_{u}(Y_{u})\notag
    \end{align}

    Therefore, even in this case, agent \({u}\) remains \(\EFXT\) towards all the agents. Since the bundles of \(s_1\) and \(s_2\) in \(Y\) are strictly better than those in \(X\), the agents \(s_1\) and \(s_2\) continue to be \(\EFXT\) towards all the other agents. Thus, \(Y\) is a \(\EFXT\) allocation. The same argument holds for $|\mathcal{C}|=1$.

    \noindent\paragraph{Case 2: $|\mathcal{C}|=2$, and only one source in $G_e(X)$: } 

    Now, consider the case when there are two critical goods in \(\mathcal{C}\) and only one source, say \(d_1\) in \(G_e(X)\). Here, the whole of \(\mathcal{C}\) is allocated to \(d_1\). We show this allocation \(Y\) to be \(\EFXT\). Let $\mathcal{C}=\{g_1,g_2\}$.
    Let \(u\) be an agent who has a critical good, say \(g_1\in \mathcal{C}\). From Properties \ref{prop:d} and \ref{prop:e}, we know that \(|X_{u}|=1\) and \(|X_{d_1}|=2\). Since \(d_1\) is a source in \(G_e(X)\), there is no red edge from \(u\) to \(d_1\), and hence \(v_{u}(X_{d_1}) < \frac{2}{3}v_{u}(X_{u})\). Furthermore, since \(X\) is obtained from the \(\ppatypes\) algorithm, Step~\ref{step3} fails on \(X\). That is, for any two goods $h_1,h_2\in\mathcal{P}(X)$, \(v_u(\{h_1,h_2\})\leq \frac{2}{3}v_{u}(X_{u})\). So, \(v_{u}(\mathcal{C}) \le \frac{2}{3}v_{u}(X_{u})\). From these two inequalities, we have \( v_{u}(X_{d_1} \cup \mathcal{C}) \le \frac{4}{3}v_{u}(X_{u}) < \frac{3}{2}v_{u}(X_{u})\). Therefore, agent \(u\) is \(\EFXT\) towards \(d_1\).

    Let \(u\) be an agent who does not have a critical good in \(\mathcal{C}\). Since \(d_1\) is a source, from \cref{lemma:ppa_types}, we know that \(|X_{d_1}|=2\). Let \(\mathcal{C}=\{g_1,g_2\}\) and \(X_{d_1}=\{g',g''\}\). Since \(G_e(X)\) is acyclic and \(d_1\) is the only source, there must exist a path from \(d_1\) to \(u\). However, since \(X\) was obtained from \(\ppatypes\), Step~\ref{step9} fails on \(X\). So, we have that for any two goods \(h_1,h_2\in X_{d_1}\cup \mathcal{C}\), \(v_{u}(X_{u})\geq v_{u}(\{h_1,h_2\})\). Therefore, for any three goods, say \(\{h_1,h_2,h_3\}\) from \(X_{d_1}\cup \mathcal{C}\), we get \(v_{u}(X_{u})\geq \frac{2}{3}v_{u}(\{h_1,h_2,h_3\})\). That is, for any good \(h\in Y_{d_1}=X_{d_1}\cup\mathcal{C}\), 
    \begin{align}
        v_{u}(X_{u})\geq \frac{2}{3}v_{u}(Y_{d_1}\setminus h)
    \end{align}

    That is, agent \(u\) is \(\EFXT\) towards \(d_1\). Since agent \(d_1\) got a more valuable bundle in \(Y\) than in \(X\), agent \(d_1\) is also \(\EFXT\) towards all the other agents. Thus, \(Y\) is a \(\EFXT\) allocation.

    \noindent\paragraph{Case 3: $|\mathcal{C}|=3$:} Finally, consider the case when there are three critical goods, say \(g_1\), \(g_2\), and \(g_3\). Without loss of generality, let \(d_1\) be the source in \(G_e(X)\). We know that \(d_1\), and hence any agent in group $\D$ has no critical goods. Furthermore, from Lemma~\ref{lemma:same-good-critical}, we know that each group can have at most one critical good. Therefore, the three critical goods must belong to three different groups. Consequently, $d_1$ must be the only source in $G_e(X)$, and every leading agent in \(\A\), \(\B\), and \(\mathcal{C}\) has a critical good. Hence they all have a singleton bundle. Having made this observation, we now consider two sub-cases based on the number of agents in group \(\D\).

    \noindent\paragraph{Case 3.1: $|\D|=1$:} If there is only one agent in group \(\D\), that is, \(d_1\), then we allocate all the three critical goods to \(d_1\). Because of \cref{prop:k-types-envy-properties}, to show that an agent, say \(a_j\in\A\) is \(\EFXT\) towards \(d_1\), it suffices to show that the corresponding leading agent, \(a_1\in \A\) is \(\EFXT\) towards \(d_1\). We will show that each of the agents \(a_1\), \(b_1\), and \(c_1\) is \(\EFXT\) towards \(d_1\). Consider agent \(a_1\). The proof is similar for agents \(b_1\) and \(c_1\). Let \(\mathcal{C}=\{g_1,g_2,g_3\}\) and \(g_1\) be the critical good for agent \(a_1\). Since \(X\) fails Step~\ref{step3}, we know that for any two goods in \(\mathcal{C}\), say $\{g,g'\}$, \(v_a(\{g,g'\})\leq \frac{2}{3}v_a(X_{a_1})\). Therefore, we have \(v_a(g_1,g_2) \le \frac{2}{3}v_a(X_{a_1})\). However, since \(g_1\) is critical for \(a_1\), we have \(v_a(g_1) > \frac{1}{2}v_a(X_{a_1})\). Combining these two inequalities, we get that \(v_a(g_2) < \frac{1}{6}v_a(X_{a_1})\). Therefore, we have:
    \begin{align}
        v_a(Y_{d_1}) &= v_a(X_{d_1} \cup \mathcal{C})\notag\\
                             &= v_a(X_{d_1}) + v_a(\{g_1,g_3\}) + v_a(g_2) \notag\\
                             &\le v_a(X_{d_1}) +  \frac{2}{3}v_a(X_{a_1}) + \frac{1}{6}v_a(X_{a_1})\tag{Since \(v_a(\{g_1,g_3\}) \le \frac{2}{3}v_a(X_{a_1})\) }\\
                             &\le \frac{2}{3}v_a(X_{a_1}) +  \frac{2}{3}v_a(X_{a_1}) + \frac{1}{6}v_a(X_{a_1})\tag{Because there is no red edge from \(a_1\) to \(d_1\)}\\
                             &= \frac{3}{2}v_a(X_{a_1})\\
                             &= \frac{3}{2}v_a(Y_{a_1})\notag
    \end{align}
    Therefore, agent \(a_1\) is \(\EFXT\) towards \(d_1\). Similarly \(b_1\) and \(c_1\) are also \(\EFXT\) towards \(d_1\). Therefore, all the agents are \(\EFXT\) towards \(d_1\), and continue to be \(\EFXT\) towards each other. Since the bundle of \(d_1\) in \(Y\) is strictly better than that in \(X\), \(d_1\) is also \(\EFXT\) towards all the other agents. Thus, \(Y\) is a \(\EFXT\) allocation.

    \noindent\paragraph{Case 3.2: $|\D|\geq 2$:} If there are at least two agents in group \(\D\), say \(d_1\) and \(d_2\), then we consider two sub-cases based on whether \(d_2\) is \(\EFXT\) towards the bundle \(X_{d_1} \cup \mathcal{C}\) or not.
    If \(d_2\) is \(\EFXT\) towards the bundle \(X_{d_1} \cup \mathcal{C}\), then we allocate all of \(\mathcal{C}\) to \(d_1\). From Case~3.1, we know that the agents in \(\A\), \(\B\), and \(\C\) are \(\EFXT\) towards \(d_1\). Since \(d_2\) is \(\EFXT\) towards \(d_1\), all the agents in \(\D\setminus \{d_1,d_2\}\) are also \(\EFXT\) towards \(d_1\), as the bundle of \(d_2\) in \(Y\) is strictly worse than that of any other agent in \(\D\setminus \{d_1,d_2\}\). Thus, \(Y\) is a \(\EFXT\) allocation.

    \paragraph{When \(d_2\) is not \(\EFXT\) towards \(X_{d_1} \cup \mathcal{C}\):} W.l.o.g let \(g_1\) be the least valuable good in \(\mathcal{C}\), according to \(v_d\). Then Algorithm~\ref{alg:few_types} allocates \(g_1\) to \(d_2\) and the remaining two critical goods, \(g_2\) and \(g_3\) to \(d_1\). We will show that this allocation \(Y\) is \(\EFXT\).

    Firstly, agents in groups \(\A\), \(\B\), and \(\C\) are \(\EFXT\) towards \(X_{d_1}\cup \mathcal{C}\) as shown in Case~3.1. Therefore, they are \(\EFXT\) towards the smaller bundle \(X_{d_1}\cup \{g_2,g_3\}\), which is now allocated \(d_1\). 
    Furthermore, from the proof of  Case~2, agents who have no critical good in \(\mathcal{C}\) are \(\EFXT\) towards \(d_1\), when \(d_1\) receives at most two goods from \(\mathcal{C}\). Hence all the agents in $\D$ are $\EFXT$ towards $d_1$. Agent \(d_1\) is \(\EFXT\) towards all of \(\N\setminus \{d_2\}\), as they retain their old bundles, and \(d_1\) receives a more valuable bundle in \(Y\) than in \(X\). It remains to prove that every agent is $\EFXT$ towards $d_2$.

    We claim that agent \(d_1\) is  \(\EFXT\) towards \(d_2\) as well. Observe that \(g_1\) is the least valuable good in \(X_{d_2} \cup \{g_1\}\).  If \(|X_{d_2}|=1\), then this observation follows from the fact that \(g_1\) is not critical for agent \(d_2\). Otherwise, if \(|X_{d_2}|=2\), then Step~\ref{step4} of \(\ppatypes\) algorithm guarantees this property. Furthermore, Step~\ref{step4} also guarantees that according to \(v_d\), \(g_1\) is less valuable than each of the goods in \(X_{d_1}\), as \(|X_{d_1}|=2\). Since \(g_1\) was chosen by \(d_2\) as the least preferred good from \(\mathcal{C}\), good \(g_1\) is also the least valuable good in \(X_{d_1} \cup \mathcal{C}\). Now, for the sake of contradiction, suppose \(d_1\) is not \(\EFXT\) towards \(d_2\). That is,
    \begin{align}
        v_d(X_{d_1} \cup \{g_2,g_3\}) < \frac{2}{3}v_d(X_{d_2})\tag{Removal of the smallest good \(g_1\) from \(Y_{d_2}\)}
    \end{align}
    But then $v_d((X_{d_1}\cup \mathcal{C})\setminus g_1)<\frac{2}{3}v_d(X_{d_2})$, contradicting the assumption that \(d_2\) is \emph{not} \(\EFXT\) towards \(X_{d_1} \cup \mathcal{C}\).

    We now show that all the other agents are \(\EFXT\) towards agent \(d_2\). For any agent $d_j, j>2$ in $\D$, if such an agent exists, $v_d(X_{d_j})\geq v_d(X_{d_2})=v_d(Y_{d_2}\setminus g_1)\geq Y_{d_2}\setminus g$ for any $g$ in $Y_{d_2}$, since $g_1$ is the least valuable good in $Y_{d_2}$. Hence each agent in $\D$ is $\EFXT$ towards $d_2$. Now consider agents from other groups. It suffices to show that the leading agents are $\EFXT$ towards $d_2$.
    
     Consider any leading agent, say \(a_1\), different from \(d_1\). Let \(g_1\) be the critical good for \(a_1\). Since \(X\) is obtained from the \(\ppatypes\) algorithm, \(X\) fails Steps~\ref{step9.1} to \ref{step9.5}. 
     \begin{enumerate}
         \item If Step~\ref{step9.1} failed, we have $|\D|=1$, and Case~3.1 above would apply. 
         \item Since $|\mathcal{C}|=3$, all leading agents except $d_1$ have critical items and hence singleton bundles. So Step~\ref{step9.2} can fail only if $|X_{d_2}|=1$ and $(d_1,d_2)\in G_e(X)$. We claim that, in this situation, $d_2$ must be $\EFXT$ towards $d_1$. Since $(d_1,d_2)\in G_e(X)$, with $|X_{d_1}|=2, |X_{d_2}|=1$, it must be the case that $v_d(X_{d_2})>\frac{3}{2}v_d(X_{d_1})$. But then, since Step~\ref{step3} fails, $v_d(X_{d_2})\geq \frac{2}{3}v_d(\{g,g'\})$ for any $g,g'\in \mathcal{P}(X)$. Therefore $v_d((X_{d_1}\cup\mathcal{C})\setminus\{g_1\})\leq \frac{2}{3}v_d(X_{d_2})+\frac{2}{3}v_d(X_{d_2})<\frac{3}{2}v_d(X_{d_2})$. Hence Step~\ref{step9.2} does not fail.
        \item Steps~\ref{step9.3}, \ref{step9.5} fail, which implies that no leading agent envies $X_{d_2}$ and no leading agent prefers an item from $d_2$ and an item from $\mathcal{P}(X)$ over their own bundle. Therefore, for any leading agent $u_1$, $v_u(X_{u_1})\geq v_u(Y_{d_2}\setminus g)$ for any $g\in Y_{d_2}$. 
     \end{enumerate}
     Thus $Y$ is a $\EFXT$ allocation.

\end{proof}

\subsection{Termination of \(\ppatypes\)}\label{sec:termination}

This section shows that the algorithm \(\ppatypes\) terminates. To do so, we first define the notion of \emph{configuration}:

\begin{definition}\label{def:configuration}
     Given an instance of at most four types of agents, a \emph{configuration} of an allocation \(X\) is a tuple \(\langle X_{a_1},\ldots,X_{a_{|\A|}},X_{b_1},\ldots,X_{c_1},\ldots,X_{d_1},\ldots,X_{d_{|\D|}} \rangle\), where the bundles of all the agents in each group are listed in an increasing order of value according to the valuation of the respective group, and the order on groups is arbitrary but fixed.
\end{definition}

\begin{lemma}\label{lem:termination}
Let $T$ be the total number of configurations. The algorithm $\ppatypes$ terminates in at most $T^2$ steps.
\end{lemma}
\begin{proof}
        By \cite{amanatidisetalPushingFrontier24a}, no agent gets the same bundle twice as long as Steps \ref{step1} to \ref{step9} of $\ppatypes$ algorithm are executed. This may not hold after Steps \ref{step9.3} and \ref{step9.5}, as the non-leading agent involved in a pseudo-cycle (referred to as agent $d_2$ in the algorithm) may get a lower valued bundle. Moreover, this agent may loose a 2-size bundle to a 1-size bundle. Thus an agent may get the same bundle more than once, but we show that the configuration before Step~\ref{step9.3} or Step~\ref{step9.5} never repeats again. This is proved in Claim~\ref{claim:repeat} below. Assuming this claim, it follows that Step~\ref{step9.3} or \ref{step9.5} can be executed at most $T$ times. 
\end{proof}
We first prove that the value of the minimum valued bundle in any group cannot keep reducing arbitrarily.

\begin{claim}\label{claim:reduction}
    Let $\mu$ be the minimum value of any bundle in a particular group at some point. Then the minimum value in that group at any later point is at least $\frac{2}{3}\mu$.
\end{claim}
\begin{proof}
    An agent exchanges a higher valued bundle for a lower valued one only in two cases: (i) a singleton bundle is exchanged for a 2-size bundle, in which case, the value of the new bundle is at least $\frac{2}{3}$ of the old bundle i.e. $\frac{2}{3}\mu$ or (ii) a non-leading agent (called $d_2$ in Steps~\ref{step9.1}-\ref{step9.5}) exchanges a 2-size bundle $X_{d_2}$ and gets a singleton bundle $Y_{d_2}$. In case (i) above, every other bundle, and hence every singleton bundle, in the group has value at least $\mu$. Therefore it cannot become lesser than $\frac{2}{3}v$. In case (ii), we have $v_d(Y_{d_2})>\frac{3}{2}v_d(X_{d_1})$ i.e. more than $\frac{3}{2}$ of the value of the leading agent's bundle. Thus the minimum value in the group does not change.
\end{proof}
\begin{claim}\label{claim:repeat}
    Let $\sigma$ be a configuration before Step~\ref{step9.3} or \ref{step9.5} is executed. Then $\sigma$ is never visited again during the course of $\ppatypes$ algorithm.
\end{claim}
\begin{proof}
    Let $\rho$ be the configuration just after Step~\ref{step9.3} or \ref{step9.5} is executed. Consider the agents $a_1$ and $d_2$ as mentioned in Steps~\ref{step9.2}, \ref{step9.3}. Let $X_{a_1}=\{g\}$ in $\sigma$. Then we claim that $a_1$ never gets the bundle $X_{a_1}$ again. 

    After Step~\ref{step9.3} or Step~\ref{step9.5}, in $\rho$, we have $v_a({Y_{a_1}})>v_a({X_{a_1}})$, and $|Y_{a_1}|=2$. From $\sigma$ to $\rho$, the minimum value in group $\A$ strictly increases. If the minimum value in group $\A$ in configuration $\sigma$ is $\mu=v_a(X_{a_1})=v_a(g)$, then by Claim~\ref{claim:reduction}, it never goes below $\frac{2}{3}\mu$. However, a singleton bundle $S$ can be received by anyone in group $\A$ only if it is more than the minimum value in the group by a factor of $\frac{3}{2}$,  i.e., $v_a(S)>\mu$. since this is not the case with $\{g\}$, group $\A$ can never receive $\{g\}$ again and hence $\sigma$ can never repeat.
\end{proof}

\section{\((1-\varepsilon)\)-EFX with \(\tilde{\mathcal{O}}(\sqrt{k/\varepsilon})\) Charity for $k$ types of agents}
This section shows that for an instance with at most \(k\) types of agents, there exists a \((1-\varepsilon)\)-EFX allocation with at most \(\tilde{\mathcal{O}}(\sqrt{k/\varepsilon})\) charity, for any \(\varepsilon\in (0,\frac{1}{2}]\).  For the reminder of this section, fix an arbitrarily chosen \(\varepsilon\in (0,\frac{1}{2}]\).

We first recall some of the definitions and Lemmas from \cite{chaudhuryImprovingEFXGuarantees.EC.2021}:
\begin{definition}[Heavy envy]
    An agent $a$ having valuation $v_a$ and holding a bundle $X_a$ is said to {\em heavily envy} a set of goods $S$ if $v_a(X_a)<(1-\varepsilon)v_a(S)$.
\end{definition}

We also use the notion of {\em valuable good} from \cite{chaudhuryImprovingEFXGuarantees.EC.2021}. A good $g$ is said to be {\em valuable} for an agent $a$ w.r.t. an allocation $X$, if $v_a(g)>\varepsilon v_a(X_a)$.
The notion of champion is defined differently in \cite{chaudhuryImprovingEFXGuarantees.EC.2021}. We call it \CGMchamp\ and recall the definition from \cite{chaudhuryImprovingEFXGuarantees.EC.2021}below:
\begin{definition}[\CGMchamp]
    Given an allocation $X$, we say that an agent $a$ is a \CGMchamp\ of a set of items $T$, if there is a set $S \subseteq T$ such that $v_a(X_a) < (1-\varepsilon)\cdot v_a(S)$ and no agent (including $a$) envies $S$ up to a factor of $(1 -\varepsilon)$, following the removal of a single good from $S$. Thus, for each agent $b$, we
have $(1-\varepsilon)\cdot v_b(S \setminus \{h\}) \leq v_b(X_b)$ for all $h \in S$. If $T=X_c\cup\{g\}$ for some agent $c$ and an unallocated good $g$, then we say that $a$ is a \CGMchamp\ for $c$ w.r.t. $g$. The set $S$ is called a {\em witness subset} for $a$ w.r.t. $T$. 
\end{definition}

We consolidate some observations from \cite{chaudhuryImprovingEFXGuarantees.EC.2021} in the following proposition and give a proof sketch for the sake of completeness

\begin{restatable}{proposition}{CGMchampvaluable}
\label{prop:CGM-champ-valuable}
    Let $X$ be a partial\(\EFXe\) allocation, $G(X)$ be its envy graph, and let \(\mathcal{P}(X)\) be the set of unallocated goods. Then the following hold:
    \begin{enumerate}
        \item\label{itm:CGM-U2} If there is an agent $i \in [n]$ that heavily envies $\mathcal{P}(X)$, then in polynomial time, we can determine a  \(\EFXe\)    allocation $Y \succ X$.
        \item\label{itm:CGM-valuable} If an unallocated good $g$ is not valuable to any agent, then $g$ can be allocated to a source $s$ in $G(X)$, to get another  \(\EFXe\)    allocation $Y$ such that $Y\succ X$. 
        \item\label{itm:CGM-val-goods-bound} If no agent heavily envies the set of unallocated goods, then the number of goods that any agent $a$ finds valuable is at most $\frac{2}{\varepsilon}$.
        \item\label{itm:CGM-U3} If there exists a set of sources $s_1,\ldots,s_t$  in $G(X)$, a set of unallocated goods $g_1, g_2, \ldots, g_t$, and a set of agents $a_1, a_2, \ldots, a_t$, such that each $a_i$ is reachable from $s_i$ in $G(X)$ and $a_i$ is the \CGMchamp\ of $X_{s_{i+1}} \cup \{g_{i+1}\}$ (indices taken modulo $t$), then a  \(\EFXe\)    allocation $Y \succ X$ can be computed in polynomial time.
    \end{enumerate}
\end{restatable}

\begin{proof}
    Part~\ref{itm:CGM-U2} holds because if an agent heavily envies $\mathcal{P}(X)$ then there is some agent $a$ who is a \CGMchamp\ of $\mathcal{P}(X)$, and $S$ be the corresponding witness subset. Exchanging $X_a$ with $S$ gives a Pareto dominant allocation.

    Part~\ref{itm:CGM-valuable} holds because no agent envies a source in $G(X)$. Thus, for any agent $a$, if $v_a(g)\leq \varepsilon v_a(X_a)$ then $v_a(X_s\cup\{g\}\leq (1+\varepsilon)v_a(X_a)$, i.e. $(1-\varepsilon)v_a(X_s\cup\{g\})\leq v_a(X_a)$. This is an allocation where $s$ is better off than that in $X$ and every other agent retains the same bundle as in $X$.

    Part~\ref{itm:CGM-val-goods-bound} follows because if an agent $a$ does not heavily envy the set of unallocated goods $\mathcal{P}(X)$, then $v_a(X_a)\geq (1-\varepsilon)v_a(\mathcal{P}(X))$ i.e. $v_a(\mathcal{P}(X))\leq \frac{1}{1-\varepsilon}v_a(X_a)\leq 2v_a(X_a)$ assuming\footnote{This assumption is valid since $\frac{1}{2}$-EFX exists even for complete allocations \cite{plautroughgardenAlmostEnvyFreeness20}.} $\varepsilon\leq \frac{1}{2}$.

    Part~\ref{itm:CGM-U3} follows from the facts that the paths in $G(X)$ from each $s_i$ to $a_i$, together with the \CGMchamp\ edges $(a_i,s_{i+1})$ (taken modulo $t$), form a cycle $C$ such that each agent $b$ gets the bundle of its successor $c$ on $C$ if the edge $(b,c)$ is an envy-edge, and gets the witness subset of $X_c\cup \{g_j\}$ if $(b,c)$ is a \CGMchamp-edge, with $b=a_j$ and $c=s_{j+1}$ for some $j$. Since the witness subsets are not strongly envied by anyone, the statement follows.
\end{proof}
\subsection{Bounded charity in terms of types of agents}
we now show an analogous result where charity is bounded in terms of the number of types of agents instead of the number of agents. This is proves as the following theorem:
\begin{theorem}\label{thm:rainbow-types}
    When there are at most $k$ types of agents with additive valuations, there exists a  \(\EFXe\)    allocation with $\tilde{O}((\frac{k}{\varepsilon})^{1/2})$ charity.
\end{theorem}
\begin{proof}
    
    Let $\mathcal{N}^1,\ldots,\mathcal{N}^k$ be the groups of agents where, each agent in $\mathcal{N}^j$, $j\in [k]$ has an additive valuation function $v_j$. Let $X$ be a partial  \(\EFXe\) allocation, where each agent gets at least one good, and $\mathcal{P}(X)$ be the set of unallocated goods. By the non-degeneracy assumption, and Proposition~\ref{prop:k-types-envy-properties}, only the leading agents (those with min. valued bundle from each group) can be the sources in $G(X)$. By non-degeneracy assumption, exactly one agent from each group is a leading agent. Let $u_1,\ldots,u_k$ be the leading agents from groups $\mathcal{N}^1,\ldots,\mathcal{N}^k$ respectively.
    
    Recall that, if a good is valuable to any agent from group $\mathcal{N}^i$, then the good is also valuable for the leading agent $u_i$. But, from Parts~\ref{itm:CGM-U2} and \ref{itm:CGM-val-goods-bound} of Proposition~\ref{prop:CGM-champ-valuable}, either a Pareto dominant  \(\EFXe\)    allocation can be obtained or any agent has at most ${2}/{\varepsilon}$ valuable goods in $\mathcal{P}(X)$. Define a good to be a {\em high-demand good} if it is valuable to more than $d$ leading agents, and {\em low-demand good} otherwise, for some $d\leq k$ to be chosen later. Then, from Part~\ref{itm:CGM-val-goods-bound} of Proposition~\ref{prop:CGM-champ-valuable}, the number of high-demand goods is at most $\frac{2k}{\varepsilon\cdot d}$.

    We bound the number of low-demand good using $R(d)$. Define the {\em group champion graph} analogous to that in \cite{chaudhuryImprovingEFXGuarantees.EC.2021}, but only on the leading agents. Thus, to each leading agent $a \in \{u_1,\ldots, u_k\}$, assign a source $s(a)$ such that $a$ is reachable from $s(a)$ in $G(X)$. If $a$ is reachable from multiple sources, assign one arbitrarily. Note that $s(a)$ must also be a leading agent. Let $\{g_1,\ldots,g_t\}$ be the set of low-demand goods. Define a $t$-partite graph $G=(\bigcup_{i=1}^t V_t, E)$ where $V_i$ contains a copy of the source $s(a)$ assigned to each leading agent $a$ that finds $g_i$ valuable. More formally, $V_i=\{(g_i,s(a))\mid a\textrm{ finds }g_i\textrm{ valuable}\}$. Thus, if the same source $s$ is assigned to two leading agents $a,b$ and both $a,b$ find $g_i$ valuable, then $V_i$ contains two copies of $s$ viz. $s(a)$ and $s(b)$. There is an edge $((g_i,s(a)),(g_j,s(b)))$ in $E$ if and only if $a$ is a \CGMchamp\ of $X_{s(b)}\cup \{g_i\}$.

    We know $|V_i|\leq d$ for each  $i\in [t]$, and each source $s$ and each good $g_i$, there is a leading agent $a$ that is a \CGMchamp\ for $s$ w.r.t. $g_i$. The source $s(a)$ assigned $a$ must be in $V_i$ since $a$ finds $g_i$ valuable by Part~\ref{itm:CGM-valuable} of Proposition~\ref{prop:CGM-champ-valuable}. This shows that each vertex $s\in V_j$ has an incoming edge from some vertex in $V_i$ for each $j,i\in [t], j\neq i$. Thus either $t\leq R(d)$ or there is a cycle $C$ in $G$ that visits each $V_i$ at most once, by Definition~\ref{def:rainbow}. By Part~\ref{itm:CGM-U3} of Proposition~\ref{prop:CGM-champ-valuable}, such a cycle gives a Pareto dominant  \(\EFXe\)    allocation. When $t\leq R(d)$, the number of unallocated goods is at most $R(d)+\frac{2k}{\varepsilon\cdot d}=O(\max\{\frac{2k}{\varepsilon\cdot d},R(d)\})$. Since $R(d)=O(d\log d)$, the theorem follows.   
\end{proof}
\section{Conclusion}\label{sec:conclusion}
We extend the existence of partial and approximate EFX to new settings, thereby adding to the state of the art knowledge on EFX. It will be interested to extend the four types $\frac{2}{3}$ EFX to 7 types of agents like that in \cite{amanatidis2021maximum}. 

\newpage
\bibliographystyle{plainnat}
\bibliography{references}

\clearpage
\appendix

{\noindent\textbf{\huge Appendix}}
\bigskip

\setcounter{secnumdepth}{2}

\section{Missing Proofs from Preliminaries}\label{sec:missing-proofs-preliminaries}

\subsection{The Non-Degeneracy Assumption}\label{sec:non-degeneracy-assumption}
In \cite{chaudhuryetalEFXExistsJ.ACM24}, it is shown that if $\EFX{}$ exists for all non-degenerate instances (Definition~\ref{def:non-degeneracy}) with \(n\) agents, then it exists for all instances of \(n\) agents. We extend this result here for approximate \(\EFX{}\), for the sake of completeness.

\begin{lemma}\label{lem:non-degeneracy}
    If for any $\alpha\in(0,1]$, if $\EFXa$ exists for all non-degenerate instances then $\EFXa$ exists for all instances.
\end{lemma}
\begin{proof}
    Let 
    \[
        \delta \coloneq \min_{i\in\mathcal{N}}\min\limits_{\substack{S,T\subseteq \mathcal{M}\\ v_i(S)\neq v_i(T)}}|\alpha\cdot v_i(S)-v_i(T)|
    \]
    
    And, Let \[0<\varepsilon<\frac{\delta}{2^{m+1}}\]
    
    Replace each valuation function $v_i$ for each  $i\in \mathcal{N}$ by a perturbed valuation function $v'_i$ defined as follows. For each ${S\subseteq \mathcal{M}}$, 
    
    \begin{align}
        v'_i(S) &\coloneq v_i(S)+\varepsilon\sum_{g_j\in S}2^j \notag
    \end{align}

    Non-degeneracy of \(v'\) follows from \cite{chaudhuryetalEFXExistsJ.ACM24}. We prove that if $X$ is an $\EFXa$ allocation according to the perturbed valuation functions $\{v'_i\}_{i\in\mathcal{N}}$ then $X$ is also an  $\EFXa$  allocation according to the given valuation functions $\{v_i\}_{i\in\mathcal{N}}$.

For the sake of contradiction, let $X$ be an  $\EFXa$  allocation according to the perturbed valuation functions but not according to the given valuation functions. Then, there exist bundles $X_i$ and $X_j$ and a good $g_k$ such that
\[\alpha\cdot v'_i(X_j\setminus g_k)\leq v'_i(X_i)\]
but
\[\alpha\cdot v_i(X_j\setminus g_k)\geq v_i(X_i) \]
    However, we have 
      \begin{align}
        \alpha\cdot v'_i(X_j\setminus g_k)-v'_i(X_i) & = \alpha\cdot v_i(X_j\setminus g_k)-v_i(X_i)\notag\\
        &+  \varepsilon\cdot\big(\alpha\cdot\sum_{g_\ell\in X_j\setminus g_k}2^\ell-\sum_{g_\ell\in X_i}2^\ell\big)\notag\\
        & \geq  \delta-\varepsilon\cdot\sum_{g_\ell\in X_i}2^\ell\notag\\
        & \geq  \delta - \varepsilon\cdot (2^{m+1}-1)\notag\\
        & >  0\notag
    \end{align}
Thus $\alpha\cdot v_i(X_j\setminus g_k)> v_i(X_i)\Rightarrow \alpha\cdot v'_i(X_j\setminus g_k)> v'_i(X_i)$, contradicting the assumption that $X$ is  $\EFXa$  with respect to the perturbed valuations.
\end{proof}
\section{More Details on the \(\ppa\) Algorithm}\label{sec:ppa-algorithm}

\citet{amanatidisetalPushingFrontier24a} present the \(\ppa\) algorithm for computing an \(\EFXT\) allocation for up to seven agents. This algorithm consists of a sequence of nice steps, each of which is executed only if the previous step fails. The algorithm terminates when the allocation failes to meet the conditions of all of the steps. The algorithm is described in Algorithm~\ref{alg:3PA_plus}. In Steps~\ref{step5} and \ref{step8}, the algorithm invokes the \textsc{All Cycles Resolution} subroutine, which is described in Algorithm~\ref{alg:all-cycles-resolution}. When an allocation \(X\) and a one of the envy graphs \(G\), \(G_\mathrm{r}\), or \(G_\mathrm{e}\) is given, the subroutine returns an updated allocation \(X'\) such that the envy graph no longer contains any cycles. In Step~\ref{step9}, the algorithm invokes the \textsc{Path Resolution$^*$} subroutine, which in turn calls the \textsc{Path Resolution} subroutine, both of which are described in Algorithms~\ref{alg:path-resolution-and-critical} and \ref{alg:path-resolution}. The subroutine takes a path \(\pi\) in the envy graph and returns an updated allocation \(X'\) such that the envy graph no longer contains the path \(\pi\). In Step~\ref{step7} of the algorithm, the \textsc{Singleton Pool} subroutine is invoked, which is described in Algorithm~\ref{alg:singleton-pool}. The subroutine takes a partial allocation \(X\) where only one good is unallocated and returns an updated allocation \(X'\).

\begin{algorithm}
\DontPrintSemicolon
\caption{$\textsc{PathResolution}(X,\tilde{G},\pi)$} \label{alg:path-resolution}
\SetKwComment{Comment}{/* }{ */}
\KwData{A partial allocation $X$, its graph $\tilde{G}(X)$, and a path $\pi = (i_1, i_2, \ldots, i_\ell)$ in $\tilde{G}(X)$}
\KwResult{An updated set of bundles $X_i$, one for every agent $i \in \{i_1, i_2, \ldots, i_{\ell-1}\}$. \vspace{3pt}} 
\For(\tcp*[f]{go through every $i$ such that $(i,j) \in \pi$ following the path}){$k \gets 1$\  to\  $\ell-1$} 
{
$X_{i_k} \gets X_{i_{k+1}}$ \tcp*[r]{assign to $i_k$ the bundle of the agent $i_{k+1}$ that she envies}
}
\Return $(X_i)_{i \in N: \exists (i, j) \in \Pi}$
\end{algorithm}

\begin{algorithm}
    \DontPrintSemicolon
    \caption{$\textsc{SingletonPool}(X)$} \label{alg:singleton-pool}
    \SetKwComment{Comment}{/* }{ */}
    \KwData{A partial allocation $X$ (of size at most $2$ with a \textit{single} unallocated good $g$, which some agent $i$---who has $|X_i|=1$ and is not a source in $G_r(X)$---values more than $2/3\cdot v_i(X_i)$). }
    \KwResult{A partial allocation $X'$, such that $g \notin \mathcal{P}(X')$ and some other good $g'$ is returned to $\mathcal{P}(X')$.\vspace{3pt}}
    Let $g$ be the only good in $\mathcal{P}(X)$ and let $i \in N$ be such that $|X_i|=1$ and $v_i(g) > \frac{2}{3}v_i(X_i)$\;
    \tcp{when \textsc{SingletonPool} is called, such $g$ and $i$ do exist and $i$ is not a source in $G_r(X)$}
    Let $\pi$ be a path in $G_r(X)$ starting from some source $s$ of $G_r(X)$ and terminating at $i$\;
    \tcp{when \textsc{SingletonPool} is called, no sources in $G_r(X)$ own only one good}
    $X \gets \textsc{PathResolution$^*$}(X, G_r, \pi)$\;
    \Return $X$
\end{algorithm}

\setcounter{algocf}{0}

\begin{algorithm}
    \DontPrintSemicolon
    \caption{${\textsc{PathResolution$^*$}}(X,\tilde{G},\pi)$} \label{alg:path-resolution-and-critical}
    \SetKwComment{Comment}{/* }{ */}
    \KwData{A partial allocation $X$, its graph $\tilde{G}(X)$, and a path $\pi = (s, \ldots, i)$ in\, $\tilde{G}(X)$ starting at a source $s$ of $\tilde{G}(X)$ with $|X_s|=2$}
    \KwResult{An updated partial allocation $X$.\vspace{3pt}}
    Let $g_s \in \argmax_{g\in X_s} v_i(g)$ \tcp*{recall that $s$ is the first and $i$ is the last vertex of $\pi$}
    Let $g_* \in \argmax_{g\in \mathcal{P}(X)} v_i(g)$\;
    $(X_i)_{i \in N: \exists (i, \ell) \in \pi}=\textsc{PathResolution}(X,\tilde{G},\pi)$\;
    $X_i \gets \{g_*,g_s\}$ \tcp*{agent $i$ gets her favorite goods from $X_s$ and $\mathcal{P}(X)$}
    $\mathcal{P}(X) \gets (\mathcal{P}(X) \cup X_s)\setminus \{g_*,g_s\}$\;
    \Return $X$
\end{algorithm}

\section{Error in \cite{amanatidisetalPushingFrontier24a} (Lemma 2.11)}\label{sec:error-in-3pa}

In \cite{amanatidisetalPushingFrontier24a}, the authors define Algorithms~\ref{alg:cycle-resolution} and \ref{alg:all-cycles-resolution}. The input graph \(\tilde{G}\) in both the subroutines can be any one of the three kinds of envy graphs, namely the envy graph \(G\), the reduced envy graph \(G_\mathrm{r}\), or the enhanced envy graph \(G_\mathrm{e}\). It is claimed that the \textsc{All Cycles Resolution} subroutine terminates in polynomial time. 

\begin{lemma}[{\cite[Lemma~2.11]{amanatidisetalPushingFrontier24a}}]\label{obs:all-cycle-resolution-running-time}
The \textsc{All Cycles Resolution} subroutine terminates in pseudo-polynomial time.
\end{lemma}
\begin{proof}{\cite[Proof of Lemma~2.11]{amanatidisetalPushingFrontier24a}}
Note that every time a cycle is resolved within the body of the while loop of the subroutine, the number of edges in the graph \(\tilde{G}(X)\) strictly decreases, compared to its previous number, before \(X\) was updated. This means that the while loop will terminate in at most \(m\) iterations. Each iteration, however, only needs \(O(n)\) steps, if \(\tilde{G}(X)\) is given, or \(O(n^2 + m)\) steps, if \(\tilde{G}(X)\) has to be constructed from scratch.  
\end{proof}

However, the above proof is incorrect. The proof assumes that the number of edges in the graph \(\tilde{G}(X)\) strictly decreases in each iteration of the while loop. This is not true when \(\tilde{G}(X)\) is an enhanced graph \(G_\mathrm{e}(X)\). The number of edges in the graph \(\tilde{G}(X)\) can remain the same or even increase in each iteration of the while loop.

Intuitively, when a cycle is resolved, an agent \(a\) who receives a bundle along a red edge gets a worse bundle. This could create new envy edges from \(a\). The new envy edges towards the agents who have bundles with more than one item do not get deleted in the enhanced graph. Furthermore, there could be a pair of agents \((i,j)\) such that \(j\) is not a source, and \(j\) being not a source is the only reason for \((i,j)\) to not be a red edge. In the cycle resolution, if the bundle of \(j\) gets transferred to some agent \(k\), such that \(k\) becomes a source, then \((i,k)\) becomes a red edge. In this way, there could be more red edges after the cycle resolution.

We construct a formal counter-example as follows. Consider an instance with \(13\) agents \(s\), \(i_1,\ldots,i_3\), \(j_1\ldots,j_4\),\ and \(k_1,\ldots,k_5\). Now, consider valuation functions and an allocation \(X\) satisfying the following conditions.
\begin{enumerate}
   \item \(|X_s|=2\);\ \(|X_{i_1}|=|X_{i_3}|=1\);\ \(|X_{i_2}|=2\);\ \(\forall \ell\in[5], |X_{k_\ell}| = 2\);\ \(\forall \ell\in[4], |X_{j_\ell}| = 1\)

   \item The only envy edges in \(G(X)\) are \((s,i_1),(i_1,i_2)\), and \((i_2,i_3)\). Furthermore, \(v_s(X_{i_1})>\frac{3}{2}v_s(X_s)\), \(v_{i_1}(X_{i_2})>\frac{3}{2}v_{i_1}(X_{i_1})\) and \(v_{i_2}(X_{i_3})>\frac{3}{2}v_{i_2}(X_{i_2})\). Therefore, all the three envy edges in \(G(X)\) also appear in the reduced envy graph \(G_\mathrm{r}(X)\) and also the enhanced envy graph \(G_\mathrm{e}(X)\).

   \item \(v_{i_3}(X_s)>\frac{2}{3}v_{i_3}(X_{i_3})\). Therefore, \((i_3,s)\) is a red edge in the enhanced envy graph \(G_\mathrm{e}(X)\).

   \item \(\forall \ell\in[5], v_{i_3}(X_{i_3})=v_{i_3}(X_{k_\ell})\).

   \item \(\forall \ell\in[4], v_{j_\ell}(X_{i_2})>\frac{2}{3}v_{j_\ell}(X_{j_\ell})\). However, \((j_\ell,i_2)\) is not a red edge in \(G_\mathrm{e}(X)\) as \(i_2\) is not a source vertex.\label{cond:not_source} 
\end{enumerate}

The enhanced envy graph \(G_\mathrm{e}(X)\) is shown in Figure~\ref{fig:fig1}. It has three envy edges and one red edge. Obtain allocation \(X'\) by applying the \textsc{Cycle Resolution} subroutine on the allocation \(X\). The enhanced envy graph \(G_\mathrm{e}(X')\) is as shown in Figure~\ref{fig:fig2}. As \(v_{i_3}(X_s) < v_{i_3}(X_{i_3})\), agent \(i_3\) envies each of \(k_1\ldots k_5\). Furthermore, as each \(|X_{k_\ell}|=2\), the edges \((i_3,k_\ell)\), \(\forall \ell\in[5]\), also appear in the reduced and hence in the enhanced envy graph \(G_\mathrm{e}(X')\). Therefore, the number of envy edges strictly increases. 

Furthermore, agent \(i_1\) is a source in \(X'\) and has the bundle \(X_{i_2}\). Due to Condition~\ref{cond:not_source}, we have five new red edges in \(X'\), from every \(j_\ell\) to \(i_1\), \(\forall \ell\in[4]\). Therefore, the number of red edges has also strictly increased.

\begin{algorithm}
\DontPrintSemicolon
\caption{\textsc{Cycle Resolution}\((X,\tilde{G},C)\)}
\label{alg:cycle-resolution}
\SetKwComment{Comment}{/* }{ */}
\KwData{A partial allocation \(X\), its graph \(\tilde{G}(X)\), and a cycle \(C\) in \(\tilde{G}(X)\)}
\KwResult{An updated partial allocation \(X\) such that the (implied) graph \(\tilde{G}(X)\) no longer contains the cycle \(C\).\vspace{3pt}}
   \(\tilde{X}\gets X\) \tcp*[r]{\(\tilde{X}\) is an auxiliary allocation}
   \For{every edge \((i,j) \in C\)}
   {
   \(X_i \gets \tilde{X}_j\) \tcp*[r]{swap the bundles backwards along the cycle}
   }
   \Return \(X\)
\end{algorithm}

\begin{algorithm}
\DontPrintSemicolon
\caption{\textsc{All Cycles Resolution}\((X,\tilde{G})\)} \label{alg:all-cycles-resolution}
\SetKwComment{Comment}{/* }{ */}
\KwData{A partial allocation \(X\) and its graph \(\tilde{G}(X)\).}
\KwResult{An updated partial allocation \(X\) such that its graph \(\tilde{G}(X)\) is acyclic.\vspace{3pt}}
   \While{there exists a cycle \(C\) in \(\tilde{G}(X)\)}
   {
   \(X=\textsc{Cycle Resolutuion}(\tilde{G}(X),C)\)
   }
   \Return \(X\)
\end{algorithm}

\begin{figure}
    \centering
    \begin{subfigure}[t]{0.48\textwidth}
        \centering
        \begin{tikzpicture}
            % Nodes
            \node[agent] (s) {\small \(s\) };
            \node[agent, left=1cm of s] (i1) {\small \(i_1\) };
            \node[agent, below=1cm of i1] (i2) {\small \(i_2\) };
            \node[agent, below=1cm of i2] (i3) {\small \(i_3\) };

            \node[agent, above left=2cm and 2cm of i1] (j1) {\small \(j_1\) };
            \node[agent, below=0.5cm of j1] (j2) {\small \(j_2\) };
            \node[agent, below=0.5cm of j2] (j3) {\small \(j_3\) };
            \node[agent, below=0.5cm of j3] (j4) {\small \(j_4\) };
        
            \node[agent, below=1cm of j4] (k1) {\small \(k_1\) };
            \node[agent, below=0.5cm of k1] (k2) {\small \(k_2\) };
            \node[agent, below=0.5cm of k2] (k3) {\small \(k_3\) };
            \node[agent, below=0.5cm of k3] (k4) {\small \(k_4\) };
            \node[agent, below=0.5cm of k4] (k5) {\small \(k_5\) };

            % Edges
            % \draw[envy] (b1) -- (c1); 
            \draw[envy] (s) to[out=135, in=35] (i1);
            \draw[envy] (i1) -- (i2);
            \draw[envy] (i2) -- (i3);
            \draw[champ] (i3) to[out=-45, in=135] (s);
        \end{tikzpicture}
    \caption{Enchanced Envy Graph \(G_\mathrm{e}(X)\), before the \textsc{Cycle Resolution} subroutine is applied.}
    \label{fig:fig1}
    \end{subfigure}
    \hfill
    \begin{subfigure}[t]{0.48\textwidth}
        \centering
        \begin{tikzpicture}
            % Nodes
            \node[agent] (s) {\small \(s\) };
            \node[agent, left=1cm of s] (i1) {\small \(i_1\) };
            \node[agent, below=1cm of i1] (i2) {\small \(i_2\) };
            \node[agent, below=1cm of i2] (i3) {\small \(i_3\) };

            \node[agent, above left=2cm and 2cm of i1] (j1) {\small \(j_1\) };
            \node[agent, below=0.5cm of j1] (j2) {\small \(j_2\) };
            \node[agent, below=0.5cm of j2] (j3) {\small \(j_3\) };
            \node[agent, below=0.5cm of j3] (j4) {\small \(j_4\) };
        
            \node[agent, below=1cm of j4] (k1) {\small \(k_1\) };
            \node[agent, below=0.5cm of k1] (k2) {\small \(k_2\) };
            \node[agent, below=0.5cm of k2] (k3) {\small \(k_3\) };
            \node[agent, below=0.5cm of k3] (k4) {\small \(k_4\) };
            \node[agent, below=0.5cm of k4] (k5) {\small \(k_5\) };

            % Edges
            \foreach \k in {1,2,3,4,5}{
                \draw[envy] (i3) -- (k\k);
            }

            \foreach \j in {1,2,3,4}{
                \draw[champ] (j\j) -- (i1);
            }

        \end{tikzpicture}
        \caption{Enhanced Envy Graph \(G_\mathrm{e}(X')\) obtained after applying the \textsc{Cycle Resolution} subroutine on \(X\).}
    \label{fig:fig2}
    \end{subfigure}
    \caption{Enhanced Envy Graphs before and after resolving the cycle $(s,i_1,i_2,i_3,s)$.}

\end{figure}
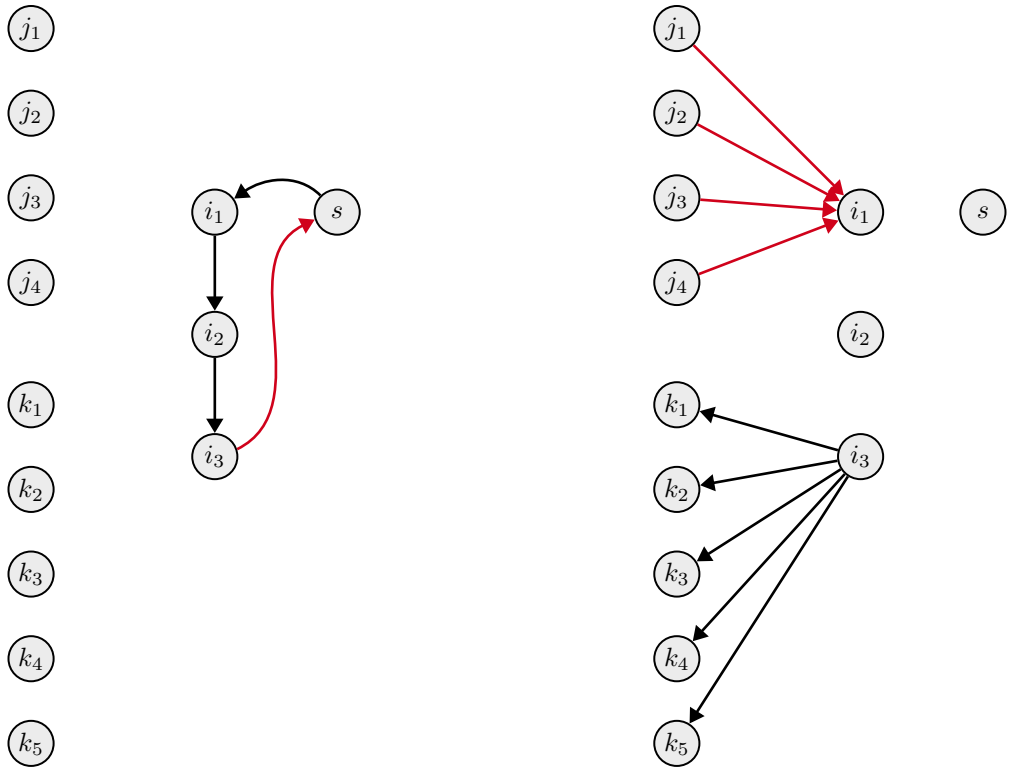

\subsection{Corrected Proof of Lemma~\ref{obs:all-cycle-resolution-running-time}}\label{sec:corrected-proof-ece}

We show that the \textsc{All Cycles Resolution} subroutine terminates in pseudo-polynomial time.

Given an allocation, define a potential function as follows.

\begin{definition}[Potential Function \(\phi\)]\label{def:potential}
    Let \(w: \mathbb{N} \to \mathbb{R}\) be a weight function defined as:
    \[
    w(x) =
    \begin{cases}
        1 & \text{if } x = 1, \\
        \frac{3}{2} & \text{if } x \geq 2.
    \end{cases}
    \]

    Given an instance \(\mathcal{I}\), and an allocation \(X\) for \(\mathcal{I}\), 
    \[
    \phi(X) = \prod_{a \in \N} w(|X_a|) \cdot v_a(X_a),
    \]
\end{definition}

\begin{proof}

    Due to \cite{liptonetalapproximatelyfair04}, it is known that the algorithm terminates when the input graph is the envy graph \(G(X)\). Therefore we only need to show that the algorithm terminates when the input graph \(\tilde{G}(X)\) is either the reduced envy graph \(G_\mathrm{r}(X)\) or the enhanced envy graph \(G_\mathrm{e}(X)\).

    Given an allocation \(X\) and its graph \(\tilde{G}(X)\), let \(X'\) be the allocation obtained after resolving one cycle \(C\) in \(\tilde{G}(X)\). We show that \(\phi(X') > \phi(X)\).

    For an agent \(a\), her contribution to the potential function is $w(|X_a|) \cdot v_a(X_a)$. Note that the agents who are not part of the cycle, retain their bundles, and hence their contributions to the potential remain unchanged. Thus, we only need to consider the agents who are part of the cycle \(C\). 

    Consider two consecutive agents \(a\) and \(b\) in the cycle \(C\). If \((a,b)\) is a red edge in \(\tilde{G}(X)\), then by the definition of the red edge, we have $|X_a|=1$, $|X_b|>1$ and $v_a(X_b)\ge \frac{2}{3}v_a(X_a)$. When the cycle is resolved, agent \(a\) receives the bundle of agent \(b\). Therefore, the new contribution of agent \(a\) to the potential function is
    \begin{align}
        w(|X_a'|) \cdot v_a(X_a') &= w(|X_b|) \cdot v_a(X_b)\notag\\
                                  &\geq \frac{3}{2} \cdot \frac{2}{3}v_a(X_a)\notag\\
                                  &\geq w(|X_a|) \cdot v_a(X_a)\notag
    \end{align}

    If \((a,b)\) is not a red edge in \(\tilde{G}(X)\), and if either $|X_a|=1$ or $|X_b|>1$, then, on receiving the bundle of agent \(b\), agent \(a\) has a new contribution to the potential function strictly increases as $w(|X_b|)\ge w(|X_a|)$ and $v_a(X_b)>v_a(X_a)$.

    Finally, if \((a,b)\) is not a red edge in \(\tilde{G}(X)\), and if both $|X_a|>1$ and $|X_b|=1$, then, by the definitions of the reduced envy graph and enhanced envy graph, we know that $v_a(X_b)>\frac{3}{2}v_a(X_a)$. Therefore, the new contribution of agent \(a\) to the potential function is $w(|X_b|) \cdot v_a(X_b) > \frac{3}{2}v_a(X_a) = w(|X_a|) \cdot v_a(X_a)$. 

    As red edges go from agents with singleton bundles to agents with bundles of size at least two, a cycle must contain at least one non red edge. Therefore, at least one agent in the cycle has a contribution to the potential function that strictly increases after the cycle resolution, while no agents contribution strictly decreases. Thus, \(\phi(X') > \phi(X)\).  Finally, since for any allocation \(X\), the maximum value of the potential is \(\phi(X)\le \max\limits_{a\in \N}\left(\frac{3}{2}\cdot v_a(\M)\right)^n \), the algorithm terminates in pseudo-polynomial time. This completes the proof.
\end{proof}

\subsection{Erratum in the   \(\ppa\) Algorithm}

In the version of the \(\ppa\) algorithm presented in \cite{amanatidisetalPushingFrontier24a}, in Step~\ref{step9} (Called as Step~{8.5} in \cite{amanatidisetalPushingFrontier24a}), the \textsc{Path Resolution$^*$} subroutine is invoked with the reduced envy graph \(G_\mathrm{r}\) as the input graph. Here, if there exists a path \(\pi\) from a source \(s\) to some agent \(i\) in \(G_r(X)\), such that $i$ prefers a good $g$ from the pool and a good $g'$ from $X_s$ to her own bundle, the swap the bundles along the envy path and give $\{g,g'\}$ to agent \(i\).

Now, in the proof of Claim~$4.9$, perticular in Case~$2$ of the proof, failure of execution of Step~$8.5$ is used to show that the allocation is \(\EFXT\). However, if the \textsc{Path Resolution$^*$} subroutine is invoked with the reduced envy graph \(G_\mathrm{r}\), then it is possible that there exists an agent $i$, a source $s$ such that $i$ prefers a good $g$ from the pool and a good $g'$ from $X_s$ to her own bundle, but there is no path from $s$ to $i$ in \(G_\mathrm{r}\), as shown in Figure~\ref{fig:fig_error_3pa}.

\begin{figure}[H]
    \centering
    \begin{center}
        \begin{tikzpicture}
            
            \node[agent] (s) {\small \(s\) };
            \node[agent, below=1cm of s] (a) {\small \(a\) };
            \node[agent, right=2cm of s] (b) {\small \(b\) };
            \node[agent, below=1cm of b] (i) {\small \(i\) };
            
            \draw[envy] (s) -- (a);
            \draw[envy] (b) -- (i);
            \draw[champ] (a) -- (b);
        \end{tikzpicture}
    \end{center}
    \caption{The path from \(s\) to \(i\) does not exist in the reduced envy graph \(G_\mathrm{r}\).}
    \label{fig:fig_error_3pa}
\end{figure}
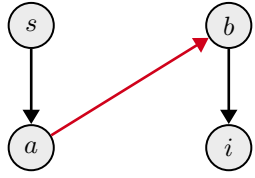

Therefore, the \textsc{Path Resolution$^*$} in Step~{8.5} must be invoked with the enhanced envy graph \(G_\mathrm{e}\). This correction is made in Algorithm~\ref{alg:3PA_plus}.

\end{document}